\documentclass[onecolumn]{IEEEtran}
\usepackage{amsmath}
\usepackage{amssymb}
\usepackage{mathrsfs}
\usepackage{cite}
\usepackage{epsfig}
\usepackage{epsfig}
\usepackage{graphics}
\usepackage{hyperref}
\usepackage{epsfig}

\newtheorem{thm}{Theorem}
\newtheorem{lem}{Lemma}
\newtheorem{coro}{Corollary}
\newtheorem{defi}{Definition}

\newtheorem{rem}{Remark}

\begin{document}

\title{The Secrecy Capacity Region of the Gaussian MIMO Broadcast Channel}

\author{Ghadamali Bagherikaram, Abolfazl S. Motahari, Amir K. Khandani\\
Coding and Signal Transmission Laboratory,\\
 Department of
Electrical
and Computer Engineering,\\
 University of Waterloo, Waterloo, Ontario,
 N2L 3G1\\
 Emails: \{gbagheri,abolfazl,khandani\}@cst.uwaterloo.ca
}
 \maketitle
\footnote{Financial support provided by Nortel and the corresponding
matching funds by the Natural Sciences and Engineering Research
Council of Canada (NSERC), and Ontario Centers of Excellence (OCE)
are gratefully acknowledged.}
\begin{abstract}
In this paper, we consider a scenario where a source node wishes to
broadcast two confidential messages for two respective receivers via
a Gaussian MIMO broadcast channel. A wire-tapper also receives the
transmitted signal via another MIMO channel. First we assumed that
the channels are degraded and the wire-tapper has the worst channel.
We establish the capacity region of this scenario. Our achievability
scheme is a combination of the superposition of Gaussian codes and
randomization within the layers which we will refer to as Secret
Superposition Coding. For the outerbound, we use the notion of
enhanced channel to show that the secret superposition of Gaussian
codes is optimal. We show that we only need to enhance the channels
of the legitimate receivers, and the channel of the eavesdropper
remains unchanged. Then we extend the result of the degraded case to
non-degraded case. We show that the secret superposition of Gaussian
codes along with successive decoding cannot work when the channels
are not degraded. we develop a Secret Dirty Paper Coding (SDPC)
scheme and show that SDPC is optimal for this channel. Finally, we
investigate practical characterizations for the specific scenario in
which the transmitter and the eavesdropper have multiple antennas,
while both intended receivers have a single antenna. We characterize
the secrecy capacity region in terms of generalized eigenvalues of
the receivers channel and the eavesdropper channel. We refer to this
configuration as the MISOME case. In high SNR we show that the
capacity region is a convex closure of two rectangular regions.
\end{abstract}
\section{Introduction}
Recently there has been significant research conducted in both
theoretical and practical aspects of wireless communication systems
with Multiple-Input Multiple-Output (MIMO) antennas. Most works have
focused on the role of MIMO in enhancing the throughput and
robustness. In this work, however, we focus on the role of such
multiple antennas in enhancing wireless security.

The information-theoretic single user secure communication problem
was first characterized by Wyner in \cite{1}. Wyner considered a
scenario in which a wire-tapper receives the transmitted signal over
a degraded channel with respect to the legitimate receiver's
channel. He measured the level of ignorance at the eavesdropper by
its equivocation and characterized the capacity-equivocation region.
Wyner's work was then extended to the general broadcast channel with
confidential messages by Csiszar and Korner \cite{2}. They
considered transmitting confidential information to the legitimate
receiver while transmitting a common information to both the
legitimate receiver and the wire-tapper. They established a
capacity-equivocation region of this channel. The secrecy capacity
for the Gaussian wire-tap channel was characterized by
Leung-Yan-Cheong in \cite{3}.

The Gaussian MIMO wire-tap channel has recently been considered by
Khisti and Wornell in \cite{4,5}. Finding the optimal distribution,
which maximizes the secrecy capacity for this channel is a nonconvex
problem. Khisti and Wornell, however, followed an indirect approach
to evaluate the secrecy capacity of  Csiszar and Korner. They used a
genie-aided upper bound and characterized the secrecy capacity as
the saddle-value of a min-max problem to show that Gaussian
distribution is optimal. Motivated by the broadcast nature of the
wireless communication systems, we considered the secure broadcast
channel with an external eavesdropper in \cite{6,6_1} and
characterized the secrecy capacity region of the degraded broadcast
channel and showed that the secret superposition coding is optimal.
Parallel and independent with our work of \cite{6,6_1}, Ekrem et.
al. in \cite{6_2,6_3} established the secrecy capacity region of the
degraded broadcast channel with an external eavesdropper. The
problem of Gaussian MIMO broadcast channel without an external
eavesdropper is also solved by Lui. et. al. in \cite{6_4,6_5,6_6}.

The capacity region of the conventional Gaussian MIMO broadcast
channel is studied in \cite{7} by Weingarten et al. The notion of an
enhanced broadcast channel is introduced in this work and is used
jointly with entropy power inequality to characterize the capacity
region of the degraded vector Gaussian broadcast channel. They
showed that the superposition of Gaussian codes is optimal for the
degraded vector Gaussian broadcast channel and that dirty-paper
coding is optimal for the nondegraded case.

In the conference version of this paper (see \cite{7_0}), we
established the secrecy capacity region of the degraded vector
Gaussian broadcast channel. Our achievability scheme, was a
combination of the superposition of Gaussian codes and randomization
within the layers which we refereed to as Secret Superposition
Coding. For the outerbound, we used the notion of enhanced channel
to show that the secret superposition of Gaussian codes is optimal.
In this paper, we aim to characterize the secrecy capacity region of
a general secure Gaussian MIMO broadcast channel. Our achievability
scheme is a combination of the dirty paper coding of Gaussian codes
and randomization within the layers. To prove the converse, we use
the notion of enhanced channel and show that the secret dirty paper
coding of Gaussian codes is optimal. We investigate practical
characterizations for the specific scenario in which the transmitter
and the eavesdropper have multiple antennas, while both intended
receivers have a single antenna. This model is motivated when a base
station wishes to broadcast secure information for small mobile
units. In this scenario small mobile units have single antenna while
the base station and the eavesdropper can afford multiple antennas.
We characterize the secrecy capacity region in terms of generalized
eigenvalues of the receivers channel and the eavesdropper channel.
We refer to this configuration as the MISOME case.In high SNR we
show that the capacity region is a convex closure of two rectangular
regions.

Parallel with our work, Ekrem et. al \cite{7_1} and Liu et. al.
\cite{7_2,7_3}, independently considered the secure MIMO broadcast
channel and established its capacity region. Ekrem et. al. used the
relationships between the minimum-mean-square-error and the mutual
information, and equivalently, the relationships between the Fisher
information and the differential entropy to provide the converse
proof. Liu et. al. considered the vector Gaussian MIMO broadcast
channel with and without an external eavesdropper. They presented a
vector generalization of Costa's Entropy Power Inequality to provide
their converse proof. In our proof, however, we enhance the channels
properly and show that the enhanced channels are proportional. We
then use the proportionality characteristic to provide the converse
proof. The rest of the paper is organized as follows. In section II
we introduce some preliminaries. In section III, we establish the
secrecy capacity region of the degraded vector Gaussian broadcast
channel. We extend our results to non-degraded and non vector case
in section IV. In Section V, we investigate the MISOME case. Section
VI concludes the paper.

\section{Preliminaries}
Consider a Secure Gaussian Multiple-Input Multiple-Output Broadcast
Channel (SGMBC) as depicted in Fig. \ref{f1}.
\begin{figure}
\centerline{\includegraphics[scale=.5]{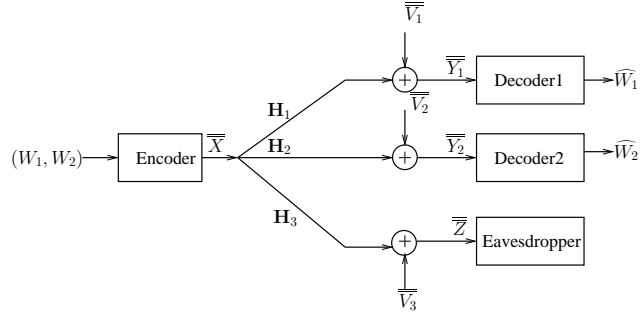}} \caption{Secure
Gaussian MIMO Broadcast Channel} \label{f1}
\end{figure}
 In this confidential setting, the transmitter wishes to send two independent messages
$(W_{1},W_{2})$ to the respective receivers in $n$ uses of the
channel and prevent the eavesdropper from having any information
about the messages. At a specific time, the signals received by the
destinations and the eavesdropper are given by
\begin{IEEEeqnarray}{lr}\label{eq1}\nonumber
\mathbf{y_{1}}=\mathbf{H}_{1}\mathbf{x}+\mathbf{n_{1}},\\
\mathbf{y_{2}}=\mathbf{H}_{2}\mathbf{x}+\mathbf{n_{2}},\\
\mathbf{z}=\mathbf{H}_{3}\mathbf{x}+\mathbf{n_{3}}, \nonumber
\end{IEEEeqnarray}
where
\begin{itemize}
  \item $\mathbf{x}$ is a real input vector of size $t\times 1$ under an input covariance constraint. We require that $E[\mathbf{x}\mathbf{x}^{T}]\preceq \mathbf{S}$ for a positive semi-definite matrix $\mathbf{S}\succeq 0$. Here,$\prec,\preceq,\succ$, and $\succeq$ represent partial ordering between symmetric matrices where $\mathbf{B}\succeq \mathbf{A}$ means that $(\mathbf{B}-\mathbf{A})$ is a positive semi-definite matrix.
  \item $\mathbf{y_{1}}$, $\mathbf{y_{2}}$, and $\mathbf{z}$ are real output vectors  which are received by the destinations and the eavesdropper respectively. These are vectors of size $r_{1} \times 1$, $r_{2} \times 1$, and $r_{3} \times 1$, respectively.
  \item $\mathbf{H}_{1}$, $\mathbf{H}_{2}$, and $\mathbf{H}_{3}$ are fixed, real gain matrices which model the channel gains between the transmitter and the receivers. These are matrices of size $r_{1} \times t$, $r_{2} \times t$, and $r_{3} \times t$ respectively. The channel state information is assumed to be known perfectly at the transmitter and at all receivers.
  \item $\mathbf{n_{1}}$, $\mathbf{n_{2}}$ and $\mathbf{n_{3}}$ are real Gaussian random vectors with zero means and covariance matrices $\mathbf{N_{1}}=E[\mathbf{n_{1}}\mathbf{n_{1}}^{T}]\succ 0$, $\mathbf{N_{2}}=E[\mathbf{n_{2}}\mathbf{n_{2}}^{T}]\succ 0$, and $\mathbf{N_{3}}=E[\mathbf{n_{3}}\mathbf{n_{3}}^{T}]\succ 0$ respectively.
\end{itemize}
Let $W_{1}$ and $W_{2}$ denote the the message indices of user $1$
and user $2$, respectively. Furthermore, let
$\overline{\overline{X}}$, $\overline{\overline{Y}}_{1}$,
$\overline{\overline{Y}}_{2}$, and $\overline{\overline{Z}}$ denote
the random channel input and random channel outputs matrices over a
block of $n$ samples. Let $\overline{\overline{V}}_{1}$,
$\overline{\overline{V}}_{2}$, and $\overline{\overline{V}}_{3}$
denote the additive noises of the channels. Thus,
\begin{IEEEeqnarray}{lr}\nonumber
\overline{\overline{Y}}_{1}=\mathbf{H}_{1}\overline{\overline{X}}+\overline{\overline{V}}_{1},\\
\overline{\overline{Y}}_{2}=\mathbf{H}_{2}\overline{\overline{X}}+\overline{\overline{V}}_{2},\\
\nonumber
\overline{\overline{Z}}=\mathbf{H}_{3}\overline{\overline{X}}+\overline{\overline{V}}_{3}.
\end{IEEEeqnarray}
Note that $\overline{\overline{V}}_{i}$ is an $r_{i} \times n$
random matrix and $\mathbf{H}_{i}$ is an $r_{i} \times t$
deterministic matrix where $i=1,2,3$. The columns of
$\overline{\overline{V}}_{i}$ are independent Gaussian random
vectors with covariance matrices $\mathbf{N_{i}}$ for $i=1,2,3$. In
addition $\overline{\overline{V}}_{i}$ is independent of
$\overline{\overline{X}}$, $W_{1}$ and $W_{2}$. A
$((2^{nR_{1}},2^{nR_{2}}),n)$ code for the above channel consists of
a stochastic encoder
\begin{equation}
f:(\{1,2,...,2^{nR_{1}}\}\times\{1,2,...,2^{nR_{2}}\})\rightarrow
\overline{\overline{\mathcal{X}}},
\end{equation}
and two decoders,
\begin{equation}
g_{1}:\overline{\overline{\mathcal{Y}}}_{1}\rightarrow
\{1,2,...,2^{nR_{1}}\},
\end{equation}
and
\begin{equation}
g_{2}:\overline{\overline{\mathcal{Y}}}_{2}\rightarrow
\{1,2,...,2^{nR_{2}}\}.
\end{equation}
where a script letter with double overline denotes the finite
alphabet of a random vector. The average probability of error is
defined as the probability that the decoded messages are not equal
to the transmitted messages; that is,
\begin{equation}
P_{e}^{(n)}=P(g_{1}(\overline{\overline{Y}}_{1})\neq W_{1}\cup
g_{2}(\overline{\overline{Y}}_{2})\neq W_{2}).
\end{equation}

The secrecy levels of confidential messages $W_{1}$ and $W_{2}$ are
measured at the eavesdropper in terms of equivocation rates, which
are defined as follows.
\begin{defi}
The equivocation rates $R_{e1}$, $R_{e2}$ and $R_{e12}$
 for the
secure broadcast channel are:
\begin{IEEEeqnarray}{lr}
R_{e1}=\frac{1}{n}H(W_{1}|\overline{\overline{Z}}),\\
\nonumber R_{e2}=\frac{1}{n}H(W_{2}|\overline{\overline{Z}}), \\
\nonumber R_{e12}=\frac{1}{n}H(W_{1},W_{2}|\overline{\overline{Z}}).
\end{IEEEeqnarray}
\end{defi}
The perfect secrecy rates $R_{1}$ and $R_{2}$ are the amount of
information that can be sent to the legitimate receivers both
reliably \emph{and} confidentially.
\begin{defi}
A secrecy rate pair $(R_{1},R_{2})$ is said to be achievable if for
any $\epsilon>0,\epsilon_{1}>0,\epsilon_{2}>0,\epsilon_{3}>0$, there
exists a sequence of $((2^{nR_{1}},2^{nR_{2}}),n)$ codes, such that
for sufficiently large $n$,
\begin{IEEEeqnarray}{rl}
\label{l0}P_{e}^{(n)}&\leq \epsilon,\\
\label{l1}
 R_{e1}&\geq R_{1}-\epsilon_{1},\\
\label{l2} R_{e2}&\geq R_{2}-\epsilon_{2},\\
\label{l3}
 R_{e12}&\geq R_{1}+R_{2}-\epsilon_{3}.
\end{IEEEeqnarray}
\end{defi}
In the above definition, the first condition concerns the
reliability, while the other conditions guarantee perfect secrecy
for each individual message and both messages as well. The model
presented in (\ref{eq1}) is SGMBC. However, we will initially
consider two subclasses of this channel and then generalize our
results for the SGMBC.

The  first subclass that we will consider is the Secure Aligned
Degraded MIMO Broadcast Channel (SADBC). The MIMO broadcast channel
of (\ref{eq1}) is said to be aligned if the number of transmit
antennas is equal to the number of receive antennas at each of the
users and the eavesdropper ($t=r_{1}=r_{2}=r_{3}$) and the gain
matrices are all identity matrices
$(\mathbf{H}_{1}=\mathbf{H}_{2}=\mathbf{H}_{3}=\mathbf{I})$.
Furthermore, if the additive noise vectors' covariance matrices are
ordered such that $0\prec \mathbf{N_{1}}\preceq
\mathbf{N_{2}}\preceq \mathbf{N_{3}}$, then the channel is SADBC.

The second subclass we consider is a generalization of the SADBC.
The MIMO broadcast channel of (\ref{eq1}) is said to be Secure
Aligned MIMO Broadcast Channel (SAMBC) if it is aligned and not
necessarily degraded. In other words, the additive noise vector
covariance matrices are not necessarily ordered. A time sample of an
SAMBC is given by the following expressions,
\begin{IEEEeqnarray}{lr}
\mathbf{y_{1}}=\mathbf{x}+\mathbf{n_{1}},\\ \nonumber
\mathbf{y_{2}}=\mathbf{x}+\mathbf{n_{2}},\\ \nonumber
\mathbf{z}=\mathbf{x}+\mathbf{n_{3}},
\end{IEEEeqnarray}
where, $\mathbf{y_{1}}$, $\mathbf{y_{2}}$, $\mathbf{z}$,
$\mathbf{x}$ are real vectors of size $t \times 1$ and
$\mathbf{n_{1}}$, $\mathbf{n_{2}}$, and $\mathbf{n_{3}}$ are
independent and real Gaussian noise vectors such that
$\mathbf{N_{i}=E[\mathbf{n_{i}}\mathbf{n_{i}}^{T}]} \succ 0_{t
\times t}$ for $i=1,2,3$.

\section{The Capacity Region of The SADBC}
In this section, we characterize the capacity region of the SADBC.
In \cite{6}, we considered the degraded broadcast channel with
confidential messages and establish its secrecy capacity region.
\begin{thm}\label{t1}
The capacity region for transmitting independent secret messages
over the degraded broadcast channel is the convex hull of the
closure of all $(R_{1},R_{2})$ satisfying
\begin{IEEEeqnarray}{rl}\label{l7}
    R_{1}&\leq I(X;Y_{1}|U)-I(X;Z|U), \\ \label{l71}
    R_{2}&\leq I(U;Y_{2})-I(U;Z).
\end{IEEEeqnarray}
for some joint distribution $P(u)P(x|u)P(y_{1}, y_{2},z|x)$.
\end{thm}
\begin{proof}
Our achievable coding scheme is based on Cover's superposition
scheme and random binning. We refer to this scheme as the Secret
Superposition Scheme. In this scheme, randomization in the first
layer increases the secrecy rate of the second layer. Our converse
proof is based on a combination of the converse proof of the
conventional degraded broadcast channel and Csiszar Lemma. Please
see \cite{6,6_1} for details.
\end{proof}
Note that evaluating (\ref{l7}) and (\ref{l71}) involves solving a
functional, nonconvex optimization problem. Usually nontrivial
techniques and strong inequalities are used to solve optimization
problems of this type. Indeed, for the single antenna case,
\cite{8,8_1} successfully evaluated the capacity expression of
(\ref{l7}) and (\ref{l71}). Liu et al. in \cite{9} evaluated the
capacity expression of MIMO wire-tap channel by using the channel
enhancement method. In the following section, we state and prove our
result for the capacity region of SADBC.

First, we define the achievable rate region due to Gaussian codebook
under a covariance matrix constraint $\mathbf{S}\succeq 0$. The
achievability scheme of Theorem \ref{t1} is the secret superposition
of Gaussian codes and successive decoding at the first receiver.
According to the above theorem, for any covariance matrix input
constraint $\mathbf{S}$ and two semi-definite  matrices
$\mathbf{B_{1}}\succeq 0$ and $\mathbf{B_{2}}\succeq 0$ such that
$\mathbf{B_{1}}+\mathbf{B_{2}} \preceq \mathbf{S}$, it is possible
to achieve the following rates,
\begin{IEEEeqnarray}{rl}\nonumber
    R_{1}^{G}(\mathbf{B_{1,2}},\mathbf{N_{1,2,3}})=\frac{1}{2}\left[\log\left|\mathbf{N_{1}^{-1}}(\mathbf{B_{1}}+\mathbf{N_{1}})\right|-\frac{1}{2}\log\left|\mathbf{N_{3}^{-1}}(\mathbf{B_{1}}+\mathbf{N_{3}})\right|\right]^{+},\\ \nonumber
   R_{2}^{G}(\mathbf{B_{1,2}},\mathbf{N_{1,2,3}})=\frac{1}{2}\left[\log\frac{\left|\mathbf{B_{1}}+\mathbf{B_{2}}+\mathbf{N_{2}}\right|}{\left|\mathbf{B_{1}}+\mathbf{N_{2}}\right|}-\frac{1}{2}\log\frac{\left|\mathbf{B_{1}}+\mathbf{B_{2}}+\mathbf{N_{3}}\right|}{\left|\mathbf{B_{1}}+\mathbf{N_{3}}\right|}\right]^{+}.
\end{IEEEeqnarray}
The Gaussian rate region of SADBC is defined as follows.
\begin{defi}
Let $\mathbf{S}$ be a positive semi-definite matrix. Then, the
Gaussian rate region of SADBC under a covariance matrix constraint
$\mathbf{S}$ is given by
\begin{IEEEeqnarray}{rl}\nonumber
\mathcal{R}^{G}(\mathbf{S},\mathbf{N_{1,2,3}})=   \left\{
                                                       \begin{array}{ll}
                                                          \left(R_{1}^{G}(\mathbf{B_{1,2}},\mathbf{N_{1,2,3}}),R_{2}^{G}(\mathbf{B_{1,2}},\mathbf{N_{1,2,3}})\right)|  \\ \hbox{s.t}~~\mathbf{S}-(\mathbf{B_{1}}+\mathbf{B_{2}})\succeq 0,~\mathbf{B_{k}}\succeq 0, ~ k=1,2
                                                       \end{array}
                                                     \right\}.
\end{IEEEeqnarray}
\end{defi}
We will show that $\mathcal{R}^{G}(\mathbf{S},\mathbf{N_{1,2,3}})$
is the capacity region of the SADBC. Before that, certain
preliminaries need to be addressed. We begin by characterizing the
boundary of the Gaussian rate region.
\begin{rem}
Note that in characterizing the capacity region of the conventional Gaussian MIMO broadcast channel
 Weingarten et al. \cite{7} proved that on the boundary of the above region we have $\mathbf{B_{1}}+\mathbf{B_{2}}=\mathbf{S}$
  which maximizes the rate $R_{2}$. In our argument, however, the boundary is not characterized with this equality as rate $R_{2}$ may decreases by increasing $\mathbf{B_{1}}+\mathbf{B_{2}}$.
\end{rem}
\begin{defi}
The rate vector $R^{*}=(R_{1},R_{2})$ is said to be an optimal
Gaussian rate vector under the covariance matrix $\mathbf{S}$, if
$R^{*}\in \mathcal{R}^{G}(\mathbf{S},\mathbf{N_{1,2,3}})$ and if
there is no other rate vector $R^{'*}=(R_{1}^{'},R_{2}^{'})\in
\mathcal{R}^{G}(\mathbf{S},\mathbf{N_{1,2,3}})$ such that
$R_{1}^{'}\geq R_{1}$ and $R_{2}^{'}\geq R_{2}$ where at least one
of the inequalities is strict. The set of positive semi-definite
matrices $(\mathbf{B_{1}^{*}},\mathbf{B_{2}^{*}})$ such that
$\mathbf{B_{1}^{*}}+\mathbf{B_{2}^{*}}\preceq \mathbf{S}$ is said to
be realizing matrices of an optimal Gaussian rate vector if the rate
vector
$\left(R_{1}^{G}(\mathbf{B_{1,2}^{*}},\mathbf{N_{1,2,3}}),R_{2}^{G}(\mathbf{B_{1,2}^{*}},\mathbf{N_{1,2,3}})\right)$
is an optimal Gaussian rate vector.
\end{defi}
In general, there is no known closed form solution for the realizing
matrices of an optimal Gaussian rate vector. Note that finding an
optimal Gaussian rate vector once again, involves solving a
nonconvex optimization problem. The realizing matrices of an optimal
Gaussian rate vector, $\mathbf{B_{1}^{*}},\mathbf{B_{2}^{*}}$ are
the solution of the following optimization problem:
\begin{IEEEeqnarray}{rl}\label{op}
&\max_{(\mathbf{B_{1}},\mathbf{B_{2}})}
R_{1}^{G}(\mathbf{B_{1,2}},\mathbf{N_{1,2,3}})+\mu
R_{2}^{G}(\mathbf{B_{1,2}},\mathbf{N_{1,2,3}})\\ \nonumber
&\hbox{s.t}~~ \mathbf{B_{1}}\succeq 0,~~~\mathbf{B_{2}}\succeq
0,~~~\mathbf{B_{1}}+\mathbf{B_{2}}\preceq\mathbf{S},
\end{IEEEeqnarray}
where $\mu \geq 1$. Next, we define a class of enhanced channel. The
enhanced channel has some fundamental properties which help us to
characterize the secrecy capacity region. We will discuss its
properties later on.
\begin{defi}
A SADBC with noise covariance matrices
$(\mathbf{N_{1}^{'}},\mathbf{N_{2}^{'}},\mathbf{N_{3}^{'}})$ is an
enhanced version of another SADBC with noise covariance matrices
$(\mathbf{N_{1}},\mathbf{N_{2}},\mathbf{N_{3}})$ if
\begin{equation}
\mathbf{N_{1}^{'}}\preceq \mathbf{N_{1}},~~\mathbf{N_{2}^{'}}\preceq
\mathbf{N_{2}},~~\mathbf{N_{3}^{'}}=
\mathbf{N_{3}},~~\mathbf{N_{1}^{'}}\preceq \mathbf{N_{2}^{'}}.
\end{equation}
\end{defi}
Obviously, the capacity region of the enhanced version contains the
capacity region of the original channel. Note that in characterizing
the capacity region of the conventional Gaussian MIMO broadcast
channel, all channels must be enhanced by reducing the noise
covariance matrices. In our scheme, however, we only enhance the
channels for the legitimate receivers and the channel of the
eavesdropper remains unchanged. This is due to the fact that the
capacity region of the enhanced channel must contain the original
capacity region. Reducing the noise covariance matrix of the
eavesdropper's channel, however, may reduce the secrecy capacity
region. The following theorem connects the definitions of the
optimal Gaussian rate vector and the enhanced channel.
\begin{thm}\label{t2}
Consider a SADBC with positive definite noise covariance matrices
$(\mathbf{N_{1}},\mathbf{N_{2}},\mathbf{N_{3}})$. Let
$\mathbf{B_{1}^{*}}$ and $\mathbf{B_{2}^{*}}$ be realizing matrices
of an optimal Gaussian rate vector under a transmit covariance
matrix constraint $\mathbf{S}\succ 0$. There then exists an enhanced
SADBC with noise covariance matrices
$(\mathbf{N_{1}^{'}},\mathbf{N_{2}^{'}},\mathbf{N_{3}^{'}})$ that
the following properties hold.
\begin{enumerate}
  \item Enhancement:\\ \nonumber
   $\mathbf{N_{1}^{'}}\preceq \mathbf{N_{1}},~~~\mathbf{N_{2}^{'}}\preceq \mathbf{N_{2}},~~~\mathbf{N_{3}^{'}}= \mathbf{N_{3}},~~~\mathbf{N_{1}^{'}}\preceq \mathbf{N_{2}^{'}}$,
 \item Proportionality:\\ \nonumber
  There exists an $\alpha \geq 0$ and a matrix $\mathbf{A}$ such that \\ \nonumber
  $(\mathbf{I}-\mathbf{A})(\mathbf{B_{1}^{*}}+\mathbf{N_{1}^{'}})=\alpha\mathbf{A}(\mathbf{B_{1}^{*}}+\mathbf{N_{3}^{'}})$,
  \item Rate and optimality preservation: \\ \nonumber
  $R_{k}^{G}(\mathbf{B_{1,2}^{*}},\mathbf{N_{1,2,3}})=R_{k}^{G}(\mathbf{B_{1,2}^{*}},\mathbf{N_{1,2,3}^{'}}) ~~~ \forall k=1,2$,  \nonumber furthermore, $\mathbf{B_{1}^{*}}$ and $\mathbf{B_{2}^{*}}$ are realizing matrices of an optimal Gaussian rate vector in the enhanced channel.
\end{enumerate}
\end{thm}

\begin{proof}
The realizing matrices $\mathbf{B_{1}^{*}}$ and $\mathbf{B_{2}^{*}}$
are the solution of the optimization problem of (\ref{op}). Using
Lagrange Multiplier method, this constraint optimization problem is
equivalent to the following unconditional optimization problem:
\begin{IEEEeqnarray}{rl}\nonumber
\max_{(\mathbf{B_{1}},\mathbf{B_{2}})}
&R_{1}^{G}(\mathbf{B_{1,2}},\mathbf{N_{1,2,3}})+\mu
R_{2}^{G}(\mathbf{B_{1,2}},\mathbf{N_{1,2,3}})+Tr\{\mathbf{B_{1}}\mathbf{O_{1}}\}\\
\nonumber
&+Tr\{\mathbf{B_{2}}\mathbf{O_{2}}\}+Tr\{\mathbf{(S-B_{1}-B_{2})}\mathbf{O_{3}}\},
\end{IEEEeqnarray}

where $\mathbf{O_{1}}$, $\mathbf{O_{2}}$, and $\mathbf{O_{3}}$ are
positive semi-definite $t \times t$ matrices such that
$Tr\{\mathbf{B_{1}^{*}}\mathbf{O_{1}}\}=0$,
$Tr\{\mathbf{B_{2}^{*}}\mathbf{O_{2}}\}=0$, and
$Tr\{\mathbf{(S-B_{1}^{*}-B_{2}^{*})}\mathbf{O_{3}}\}=0$. As all
$\mathbf{B_{k}^{*}}, ~k=1,2$, $\mathbf{O_{i}},~ i=1,2,3$, and
$\mathbf{S-B_{1}^{*}-B_{2}^{*}}$ are positive semi-definite
matrices, then we must have
$\mathbf{B_{k}^{*}}\mathbf{O_{k}}=0,~~k=1,2$ and
$(\mathbf{S-B_{1}^{*}-B_{2}^{*}})\mathbf{O_{3}}=0$. According to the
necessary KKT conditions, and after some manipulations we have:
\begin{IEEEeqnarray}{rl}\label{eq3}
(\mathbf{B_{1}^{*}}+\mathbf{N_{1}})^{-1}+(\mu-1)(\mathbf{B_{1}^{*}}+\mathbf{N_{3}})^{-1}+\mathbf{O_{1}}=
\mu(\mathbf{B_{1}^{*}}+\mathbf{N_{2}})^{-1}+\mathbf{O_{2}},\\\label{eq3_1}
\mu(\mathbf{B_{1}^{*}}+\mathbf{B_{2}^{*}}+\mathbf{N_{2}})^{-1}+\mathbf{O_{2}}=
\mu(\mathbf{B_{1}^{*}}+\mathbf{B_{2}^{*}}+\mathbf{N_{3}})^{-1}+\mathbf{O_{3}}.
\end{IEEEeqnarray}
We choose the noise covariance matrices of the enhanced SADBC as the
following:
\begin{IEEEeqnarray}{rl}\label{en}
\mathbf{N_{1}^{'}}&=\left(\mathbf{N_{1}}^{-1}+\mathbf{O_{1}}\right)^{-1},\\
\nonumber
\mathbf{N_{2}^{'}}&=\left(\left(\mathbf{B_{1}^{*}}+\mathbf{N_{2}}\right)^{-1}+\frac{1}{\mu}\mathbf{O_{2}}\right)^{-1}-\mathbf{B_{1}^{*}},\\
\nonumber \mathbf{N_{3}^{'}}&=\mathbf{N_{3}}.
\end{IEEEeqnarray}
As $\mathbf{O_{1}}\succeq 0$ and $\mathbf{O_{2}}\succeq 0$, then the
above choice has the enhancement property. Note that 
\begin{IEEEeqnarray}{rl}\label{eq8}
\left(\left(\mathbf{B_{1}^{*}}+\mathbf{N_{1}}\right)^{-1}+\mathbf{O_{1}}\right)^{-1}&=
\left(\left(\mathbf{B_{1}^{*}}+\mathbf{N_{1}}\right)^{-1}\left(\mathbf{I}+\left(\mathbf{B_{1}^{*}}+\mathbf{N_{1}}\right)\mathbf{O_{1}}\right)\right)^{-1}\\
\nonumber&\stackrel{(a)}{=}\left(\mathbf{I}+\mathbf{N_{1}}\mathbf{O_{1}}\right)^{-1}\left(\mathbf{B_{1}^{*}}+\mathbf{N_{1}}\right)-\mathbf{B_{1}^{*}}+\mathbf{B_{1}^{*}}\\
\nonumber&=\left(\mathbf{I}+\mathbf{N_{1}}\mathbf{O_{1}}\right)^{-1}\left(\left(\mathbf{B_{1}^{*}}+\mathbf{N_{1}}\right)-\left(\mathbf{I}+\mathbf{N_{1}}\mathbf{O_{1}}\right)\mathbf{B_{1}^{*}}\right)+\mathbf{B_{1}^{*}}\\
\nonumber&\stackrel{(b)}{=}\left(\mathbf{I}+\mathbf{N_{1}}\mathbf{O_{1}}\right)^{-1}\mathbf{N_{1}}+\mathbf{B_{1}^{*}}\\
\nonumber&=\left(\mathbf{N_{1}}\left(\mathbf{N_{1}^{-1}}+\mathbf{O_{1}}\right)\right)^{-1}\mathbf{N_{1}}+\mathbf{B_{1}^{*}}\\
\nonumber&=\left(\mathbf{N_{1}^{-1}}+\mathbf{O_{1}}\right)^{-1}+\mathbf{B_{1}^{*}}\\
\nonumber &=\mathbf{B_{1}^{*}}+\mathbf{N_{1}^{'}},
\end{IEEEeqnarray}
where $(a)$ and $(b)$ follows from the fact that
$\mathbf{B_{1}^{*}}\mathbf{O_{1}}=0$. Therefore, according to
(\ref{eq3}) the following property holds for the enhanced channel.
\begin{IEEEeqnarray}{rl}\label{eq4}\nonumber
     (\mathbf{B_{1}^{*}}+\mathbf{N_{1}^{'}})^{-1}+(\mu-1)(\mathbf{B_{1}^{*}}+\mathbf{N_{3}^{'}})^{-1}= \mu(\mathbf{B_{1}^{*}}+\mathbf{N_{2}^{'}})^{-1}.
\end{IEEEeqnarray}
Since $\mathbf{N_{1}^{'}}\preceq \mathbf{N_{2}^{'}}\preceq
\mathbf{N_{3}^{'}}$ then, there exists a matrix $\mathbf{A}$ such
that
$\mathbf{N_{2}^{'}}=(\mathbf{I}-\mathbf{A})\mathbf{N_{1}^{'}}+\mathbf{A}\mathbf{N_{3}^{'}}$
where
$\mathbf{A}=(\mathbf{N_{2}^{'}}-\mathbf{N_{1}^{'}})(\mathbf{N_{3}^{'}}-\mathbf{N_{1}^{'}})^{-1}$.
Therefore, the above equation can be written as.
\begin{IEEEeqnarray}{rl}\nonumber
     &(\mathbf{B_{1}^{*}}+\mathbf{N_{1}^{'}})^{-1}+(\mu-1)(\mathbf{B_{1}^{*}}+\mathbf{N_{3}^{'}})^{-1}=\\ \nonumber&
     \mu\left[(\mathbf{I}-\mathbf{A})(\mathbf{B_{1}^{*}}+\mathbf{N_{1}^{'}})+\mathbf{A}(\mathbf{B_{1}^{*}}+\mathbf{N_{3}^{'}})\right]^{-1}.
\end{IEEEeqnarray}
Let
$(\mathbf{I}-\mathbf{A})(\mathbf{B_{1}^{*}}+\mathbf{N_{1}^{'}})=\alpha\mathbf{A}(\mathbf{B_{1}^{*}}+\mathbf{N_{3}^{'}})$
then after some manipulations, the above equation becomes
\begin{IEEEeqnarray}{rl}
     \frac{1}{\alpha}\mathbf{I}+(\mu-1-\frac{1}{\alpha})\mathbf{A}=\frac{\mu}{\alpha+1}\mathbf{I}.
\end{IEEEeqnarray}
The above equation is satisfied by $\alpha=\frac{1}{\mu-1}$ which
completes the proportionality property. We can now prove the rate
conservation property. The expression
$\frac{\left|\mathbf{B_{1}^{*}}+\mathbf{N_{1}^{'}}\right|}{\left|\mathbf{N_{1}^{'}}\right|}$
can be written as follow.
\begin{IEEEeqnarray}{rl}\label{eq5}
\frac{\left|\mathbf{B_{1}^{*}}+\mathbf{N_{1}^{'}}\right|}{\left|\mathbf{N_{1}^{'}}\right|}&=\frac{\left|\mathbf{I}\right|}{\left|\mathbf{N_{1}^{'}}\left(\mathbf{B_{1}^{*}}+\mathbf{N_{1}^{'}}\right)^{-1}\right|}\\
\nonumber &=\frac{\left|\mathbf{I}\right|}{\left|\left(\mathbf{B_{1}^{*}}+\mathbf{N_{1}^{'}}-\mathbf{B_{1}^{*}}\right)\left(\mathbf{B_{1}^{*}}+\mathbf{N_{1}^{'}}\right)^{-1}\right|}\\
\nonumber &=\frac{\left|\mathbf{I}\right|}{\left|\mathbf{I}-\mathbf{B_{1}^{*}}\left(\mathbf{B_{1}^{*}}+\mathbf{N_{1}^{'}}\right)^{-1}\right|}\\
\nonumber &=\frac{\left|\mathbf{I}\right|}{\left|\mathbf{I}-\mathbf{B_{1}^{*}}\left((\mathbf{B_{1}^{*}}+\mathbf{N_{1}})^{-1}+\mathbf{O_{1}}\right)\right|}\\
\nonumber &\stackrel{(a)}{=}\frac{\left|\mathbf{I}\right|}{\left|\mathbf{I}-\mathbf{B_{1}^{*}}\left(\mathbf{B_{1}^{*}}+\mathbf{N_{1}}\right)^{-1}\right|}\\
\nonumber
&=\frac{|\mathbf{B_{1}^{*}}+\mathbf{N_{1}}|}{|\mathbf{N_{1}}|},
\end{IEEEeqnarray}
where $(a)$ once again follows from the fact that
$\mathbf{B_{1}^{*}}\mathbf{O_{1}}=0$. To complete the proof of rate
conservation, consider the following equalities.
\begin{IEEEeqnarray}{rl}\label{eq6}
\frac{\left|\mathbf{B_{1}^{*}}+\mathbf{B_{2}^{*}}+\mathbf{N_{2}^{'}}\right|}{\left|\mathbf{B_{1}^{*}}+\mathbf{N_{2}^{'}}\right|}&=\frac{\left|\mathbf{B_{2}^{*}}\left(\mathbf{B_{1}^{*}}+\mathbf{N_{2}^{'}}\right)^{-1}+\mathbf{I}\right|}{\left|\mathbf{I}\right|}\\
\nonumber
&=\frac{\left|\mathbf{B_{2}^{*}}\left(\left(\mathbf{B_{1}^{*}}+\mathbf{N_{2}}\right)^{-1}+\frac{1}{\mu}\mathbf{O_{2}}\right)+\mathbf{I}\right|}{\left|\mathbf{I}\right|}\\
\nonumber
&\stackrel{(a)}{=}\frac{\left|\mathbf{B_{1}^{*}}+\mathbf{B_{2}^{*}}+\mathbf{N_{2}}\right|}{\left|\mathbf{B_{1}^{*}}+\mathbf{N_{2}}\right|},
\end{IEEEeqnarray}
where\ $(a)$ follows from the fact
$\mathbf{B_{2}^{*}}\mathbf{O_{2}}=0$. Therefore, according to
(\ref{eq5}), (\ref{eq6}), and the fact that
$\mathbf{N_{3}^{'}}=\mathbf{N_{3}}$, the rate preservation property
holds for the enhanced channel. To prove the optimality
preservation, we need to show that
$(\mathbf{B_{1}^{*}},\mathbf{B_{2}^{*}})$ are also realizing
matrices of an optimal Gaussian rate vector in the enhanced channel.
For that purpose, we show that the necessary KKT conditions for the
enhanced channel coincides with the KKT conditions of the original
channel. The expression
$\mu(\mathbf{B_{1}^{*}}+\mathbf{B_{2}^{*}}+\mathbf{N_{2}^{'}})^{-1}$
can be written as follows
\begin{IEEEeqnarray}{rl}
\mu\left(\mathbf{B_{1}^{*}}+\mathbf{B_{2}^{*}}+\mathbf{N_{2}^{'}}\right)^{-1}&\stackrel{(a)}{=}\mu\left(\mathbf{B_{1}^{*}}+\mathbf{B_{2}^{*}}+\left(\mathbf{N_{2}}^{-1}+\frac{1}{\mu}\mathbf{O_{2}}\right)^{-1}\right)^{-1}\\
\nonumber
&=\mu\left(\mathbf{B_{1}^{*}}+\mathbf{B_{2}^{*}}\left(\mathbf{I}+\mathbf{B_{2}^{*}}^{-1}\left(\mathbf{N_{2}}^{-1}+\frac{1}{\mu}\mathbf{O_{2}}\right)^{-1}\right)\right)^{-1}\\
\nonumber &=\mu\left(\mathbf{B_{1}^{*}}+\mathbf{B_{2}^{*}}\left(\mathbf{I}+\left(\left(\mathbf{N_{2}}^{-1}+\frac{1}{\mu}\mathbf{O_{2}}\right)\mathbf{B_{2}^{*}}\right)^{-1}\right)\right)^{-1}\\
\nonumber &\stackrel{(b)}{=}\mu\left(\mathbf{B_{1}^{*}}+\mathbf{B_{2}^{*}}\left(\mathbf{I}+\left(\mathbf{N_{2}}^{-1}\mathbf{B_{2}^{*}}\right)^{-1}\right)\right)^{-1}\\
\nonumber &=\mu\left(\mathbf{B_{1}^{*}}+\mathbf{B_{2}^{*}}\left(\mathbf{I}+\mathbf{B_{2}^{*}}^{-1}\mathbf{N_{2}}\right)\right)^{-1}\\
\nonumber
&=\mu\left(\mathbf{B_{1}^{*}}+\mathbf{B_{2}^{*}}+\mathbf{N_{2}}\right)^{-1}
\end{IEEEeqnarray}
where $(a)$ follows from the definition of $\mathbf{N_{2}^{'}}$ and
$(b)$ follows from the fact that
$\mathbf{B_{2}^{*}}\mathbf{O_{2}}=0$. Therefore, according to
(\ref{eq8}), and the above equation, the KKT conditions of
(\ref{eq3}) and (\ref{eq3_1}) for the original channel can be
written as follows for the enhanced channel.
\begin{IEEEeqnarray}{rl}\label{kkt}
(\mathbf{B_{1}^{*}}+\mathbf{N_{1}^{'}})^{-1}+(\mu-1)(\mathbf{B_{1}^{*}}+\mathbf{N_{3}^{'}})^{-1}=
\mu(\mathbf{B_{1}^{*}}+\mathbf{N_{2}^{'}})^{-1},\\\label{kkt2}
\mu(\mathbf{B_{1}^{*}}+\mathbf{B_{2}^{*}}+\mathbf{N_{2}^{'}})^{-1}=
\mu(\mathbf{B_{1}^{*}}+\mathbf{B_{2}^{*}}+\mathbf{N_{3}^{'}})^{-1}+\mathbf{O_{3}}-\mathbf{O_{2}}.
\end{IEEEeqnarray}
where $\mathbf{O_{3}}-\mathbf{O_{2}}\succeq 0$. Therefore,
$R_{1}^{G}(\mathbf{B_{1,2}},\mathbf{N_{1,2,3}^{'}})+\mu
R_{2}^{G}(\mathbf{B_{1,2}},\mathbf{N_{1,2,3}^{'}})$ is maximized
when $\mathbf{B_{k}}=\mathbf{B_{k}^{*}}$ for $k=1,2$.
\end{proof}
We can now use Theorem \ref{t2} to prove that
$\mathcal{R}^{G}(\mathbf{S},\mathbf{N_{1,2,3}})$ is the capacity
region of the SADBC. We follow Bergman's approach \cite{10} to prove
a contradiction. Note that since the original channel is not
proportional, we cannot apply  Bergman's proof on the original
channel directly. Here we apply his proof on the enhanced channel
instead.
\begin{thm}\label{t3}
Consider a SADBC with positive definite noise covariance matrices
$(\mathbf{N_{1}},\mathbf{N_{2}},\mathbf{N_{3}})$. Let
$\mathcal{C}(\mathbf{S},\mathbf{N_{1,2,3}})$ denote the capacity
region of the SADBC under a covariance matrix constraint
$\mathbf{S}\succ 0$ .Then,
$\mathcal{C}(\mathbf{S},\mathbf{N_{1,2,3}})=\mathcal{R}^{G}(\mathbf{S},\mathbf{N_{1,2,3}})$.
\end{thm}
\begin{proof}
The achievability scheme is secret superposition coding with
Gaussian codebook. For the converse proof, we use a contradiction
argument and assume that there exists an achievable rate vector
$\bar{R}=(R_{1},R_{2})$ which is not in the Gaussian region. We can
apply the steps of Bergman's proof of \cite{10} on the enhanced
channel to show that this assumption is impossible. Since
$\bar{R}\notin\mathcal{R}^{G}(\mathbf{S},\mathbf{N_{1,2,3}})$, there
exist realizing matrices of an optimal Gaussian rate vector
$\mathbf{B_{1}^{*}}, \mathbf{B_{2}^{*}}$ such that
\begin{IEEEeqnarray}{rl}
R_{1}&\geq R_{1}^{G}(\mathbf{B_{1,2}^{*}},\mathbf{N_{1,2,3}}), \\
\nonumber R_{2}&\geq
R_{2}^{G}(\mathbf{B_{1,2}^{*}},\mathbf{N_{1,2,3}})+b,
\end{IEEEeqnarray}
for some $b>0$. We know by Theorem \ref{t2} that for every set of
realizing matrices of an optimal Gaussian rate vector
$\mathbf{B_{1}^{*}},\mathbf{B_{2}^{*}}$, there exists an enhanced
SADBC with noise covariance matrices
$\mathbf{N_{1}^{'}},\mathbf{N_{2}^{'}}$, such that the
proportionality and rate preservation properties hold. According to
the rate preservation property, we have
$R_{k}^{G}(\mathbf{B_{1,2}^{*}},\mathbf{N_{1,2}})=R_{k}^{G}(\mathbf{B_{1,2}^{*}},\mathbf{N_{1,2}^{'}}),~~k=1,2$.
Therefore, the preceding expression can be rewritten as follows:
\begin{IEEEeqnarray}{rl}\label{eq10}
R_{1}&\geq R_{1}^{G}(\mathbf{B_{1,2}^{*}},\mathbf{N_{1,2,3}})=R_{1}^{G}(\mathbf{B_{1,2}^{*}},\mathbf{N_{1,2,3}^{'}}), \\
\nonumber R_{2}&\geq
R_{2}^{G}(\mathbf{B_{1,2}^{*}},\mathbf{N_{1,2,3}})+b=R_{2}^{G}(\mathbf{B_{1,2}^{*}},\mathbf{N_{1,2,3}^{'}})+b,
\end{IEEEeqnarray}
According to the Theorem \ref{t1}, $R_{1}$ and $R_{2}$ are bounded
as follows:
\begin{IEEEeqnarray}{rl}\nonumber
R_{1}&\leq h(\mathbf{y_{1}}|\mathbf{u})-h(\mathbf{z}|\mathbf{u})-\left(h(\mathbf{y_{1}}|\mathbf{x},\mathbf{u})-h(\mathbf{z}|\mathbf{x},\mathbf{u})\right)\\
\nonumber R_{2}&\leq
h(\mathbf{y_{2}})-h(\mathbf{z})-\left(h(\mathbf{y_{2}}|\mathbf{u})-h(\mathbf{z}|\mathbf{u})\right)
\end{IEEEeqnarray}
Let $\mathbf{y_{1}^{'}}$ and $\mathbf{y_{2}^{'}}$  denote the
enhanced channel outputs of each of the receiving users. As
$\mathbf{u}\rightarrow\mathbf{y_{k}^{'}}\rightarrow\mathbf{y_{k}}$
forms a Markov chain for $k=1,2$ and $\mathbf{z^{'}}=\mathbf{z}$,
then we can use the data processing inequality to rewrite the above
region as follows:
\begin{IEEEeqnarray}{rl}\label{eq11}
R_{1}&\leq h(\mathbf{y_{1}^{'}}|\mathbf{u})-h(\mathbf{z}^{'}|\mathbf{u})-\left(h(\mathbf{y_{1}^{'}}|\mathbf{x},\mathbf{u})-h(\mathbf{z}^{'}|\mathbf{x},\mathbf{u})\right)\\
\nonumber&=
h(\mathbf{y_{1}^{'}}|\mathbf{u})-h(\mathbf{z}^{'}|\mathbf{u})-\frac{1}{2}\left(\log|\mathbf{N_{1}^{'}}|-\log|\mathbf{N_{3}^{'}}|)\right)\\
\label{eq13} R_{2}&\leq
h(\mathbf{y_{2}^{'}})-h(\mathbf{z}^{'})-\left(h(\mathbf{y_{2}^{'}}|\mathbf{u})-h(\mathbf{z}^{'}|\mathbf{u})\right)
\end{IEEEeqnarray}
Now, the inequalities of (\ref{eq10}) and (\ref{eq11}) have shifted
to the enhanced channel.

Since $R_{1}>R_{1}^{G}(\mathbf{B_{1,2}},\mathbf{N_{1,2,3}^{'}})$,
the inequality (\ref{eq11}) means that
\begin{IEEEeqnarray}{rl}\nonumber
h(\mathbf{y_{1}^{'}}|\mathbf{u})-h(\mathbf{z^{'}}|\mathbf{u})>\frac{1}{2}\left(\log|\mathbf{B_{1}^{*}}+\mathbf{N_{1}^{'}}|-\log|\mathbf{B_{1}^{*}}+\mathbf{N_{3}^{'}}|)\right)
\end{IEEEeqnarray}
By the definition of matrix $\mathbf{A}$ and since
$\mathbf{y_{1}^{'}}\rightarrow\mathbf{y_{2}^{'}}\rightarrow\mathbf{z^{'}}$
forms a Morkov chain, the received signals $\mathbf{z^{'}}$ and
$\mathbf{y_{2}^{'}}$ can be written as
$\mathbf{z^{'}}=\mathbf{y_{1}^{'}}+\mathbf{\widetilde{n}}$ and
$\mathbf{y_{2}^{'}}=\mathbf{y_{1}^{'}}+\mathbf{A}^{\frac{1}{2}}\mathbf{\widetilde{n}}$
where $\mathbf{\widetilde{n}}$ is an independent Gaussian noise with
covariance matrix
$\mathbf{\widetilde{N}}=\mathbf{N_{3}^{'}}-\mathbf{N_{1}^{'}}$.
According to Costa's Entropy Power Inequality and the previous
inequality, we have
\begin{IEEEeqnarray}{rl}\label{eq7}
\nonumber
h(\mathbf{y_{2}^{'}}|\mathbf{u})-&h(\mathbf{z^{'}}|\mathbf{u})
 \\ \nonumber &\geq\frac{t}{2}\log\left(|\mathbf{I}-\mathbf{A}|^{\frac{1}{t}}2^{\frac{2}{t}\left(h(\mathbf{y_{1}^{'}|\mathbf{u}})-h(\mathbf{z|^{'}\mathbf{u}})\right)}+|\mathbf{A}|^{\frac{1}{t}})\right)\\
\nonumber &>\frac{t}{2}\log\left(\frac{|\mathbf{I}-\mathbf{A}|^{\frac{1}{t}}|\mathbf{B_{1}^{*}}+\mathbf{N_{1}^{'}}|^{\frac{1}{t}}}{|\mathbf{B_{1}^{*}}+\mathbf{N_{3}^{'}}|^{\frac{1}{t}}}+|\mathbf{A}|^{\frac{1}{t}})\right)\\
&\stackrel{(a)}{=}\frac{1}{2}\log(\mathbf{B_{1}^{*}}+\mathbf{N_{2}^{'}})-\frac{1}{2}\log(\mathbf{B_{1}^{*}}+\mathbf{N_{3}^{'}})
\end{IEEEeqnarray}
where (a) is due to the proportionality property. Using (\ref{eq13})
and the fact that
$R_{2}>R_{2}^{G}(\mathbf{B_{1,2}},\mathbf{N_{1,2,3}^{'}})$, the
inequality (\ref{eq13}) means that
\begin{IEEEeqnarray}{rl}\nonumber
&h(\mathbf{y_{2}^{'}})-h(\mathbf{z^{'}})\geq R_{2}+h(\mathbf{y_{2}^{'}}|\mathbf{u})-h(\mathbf{z^{'}}|\mathbf{u})>\\
\nonumber
&\frac{1}{2}\log(\mathbf{B_{1}^{*}}+\mathbf{B_{2}^{*}}+\mathbf{N_{2}^{'}})-\frac{1}{2}\log(\mathbf{B_{1}^{*}}+\mathbf{B_{2}^{*}}+\mathbf{N_{3}^{'}})
\end{IEEEeqnarray}
On the other hand, Gaussian distribution maximizes
$h(\mathbf{x}+\mathbf{n_{2}})-h(\mathbf{x}+\mathbf{n_{3}})$ (See
\cite{11}) and $(\mathbf{B_{1}^{*}},\mathbf{B_{2}^{*}})$ satisfying
the KKT conditions of (\ref{kkt2}). Therefore, the above inequality
is a contradiction.
\end{proof}

\section{The Capacity Region of the SAMBC}
In this section, we characterize the secrecy capacity region of the
aligned (but not necessarily degraded) MIMO broadcast channel. Note
that since the SAMBC is not degraded, there is no  single-letter
formula for its capacity region. In addition, the secret
superposition of Gaussian codes along with successive decoding
cannot work when the channel is not degraded. In \cite{6}, we
presented an achievable rate region for the general secure Broadcast
channel. Our achievable coding scheme is based on a combination of
the random binning and the Gelfand-Pinsker binning schemes. We first
review this scheme and then based on this result, we develop an
achievable secret coding scheme for the SAMBC. After that, based on
the Theorem \ref{t2}, we provide a full characterization of the
capacity region of SAMBC.

\subsection{Secret Dirty-Paper Coding Scheme and Achievability Proof}
In \cite{6}, we established an achievable rate region for the
general secure broadcast channel. This scheme enables both joint
encoding at the transmitter by using Gelfand-Pinsker binning and
preserving confidentiality by using random binning. The following
theorem summarizes the encoding strategy. The confidentiality proof
is given in Appendix II for completeness.
\begin{thm}:\label{t4}
Let $V_{1}$ and $V_{2}$ be auxiliary random variables and $\Omega$
be the class of joint probability densities
$P(v_{1},v_{2}x,y_{1},y_{2},z)$ that factors as
$P(v_{1},v_{2})P(x|v_{1},v_{2})P(y_{1},y_{2},z|x)$. Let
$\mathcal{R}_{I}(\pi)$ denote the union of all non-negative rate
pairs $(R_{1},R_{2})$ satisfying
 \begin{IEEEeqnarray}{rl}\nonumber
    R_{1}&\leq I(V_{1};Y_{1})-I(V_{1};Z), \\ \nonumber
    R_{2}&\leq I(V_{2};Y_{2})-I(V_{2};Z), \\ \nonumber
    R_{1}+R_{2}&\leq
    I(V_{1};Y_{1})+I(V_{2};Y_{2})-I(V_{1},V_{2};Z)-I(V_{1};V_{2}),
  \end{IEEEeqnarray}
for a given joint probability density $\pi \in\Omega$. For the
general broadcast channel with confidential messages, the following
region is achievable.
\begin{equation}\label{eq14}
\mathcal{R}_{I}=conv\left\{\bigcup_{\pi\in\Omega}\mathcal{R}_{I}(\pi)\right\}
\end{equation}
where $conv$ is the convex closure operator.
\end{thm}
\begin{rem}
If we remove the secrecy constraints by removing the eavesdropper,
then the above rate region becomes Marton's achievable region for
the general broadcast channel.
\end{rem}
\begin{proof}1) \textit{Codebook Generation}:
 The structure of the encoder is
depicted in Fig.\ref{f2}.
\begin{figure}
\centerline{\includegraphics[scale=.5]{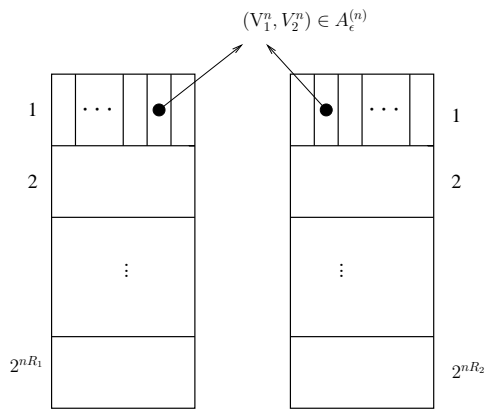}} \caption{The
Stochastic Encoder} \label{f2}
\end{figure}
Fix $P(v_{1})$, $P(v_{2})$ and  $P(x|v_{1},v_{2})$. The stochastic
encoder generates $2^{n(I(V_{1};Y_{1})-\epsilon)}$ independent and
identically distributed sequences $v_{1}^{n}$ according to the
distribution $P(v_{1}^{n})=\prod_{i=1}^{n}P(v_{1,i})$. Next,
randomly distribute these sequences into $2^{nR_{1}}$ bins such that
each bin contains $2^{n(I(V_{1};Z)-\epsilon)}$ codewords. Similarly,
it generates $2^{n(I(V_{2};Y_{2})-\epsilon)}$ independent and
identically distributed sequences $v_{2}^{n}$ according to the
distribution $P(v_{2}^{n})=\prod_{i=1}^{n}P(v_{2,i})$. Randomly
distribute these sequences into $2^{nR_{2}}$ bins such that each bin
contains $2^{n(I(V_{2};Z)-\epsilon)}$ codewords. Index each of the
above bins by $w_{1}\in\{1,2,...,2^{nR_{1}}\}$ and
$w_{2}\in\{1,2,...,2^{nR_{2}}\}$ respectively.

2) \textit{Encoding}: To send messages $w_{1}$ and $w_{2}$, the
transmitter looks for $v_{1}^{n}$ in bin $w_{1}$ of the first bin
set and looks for $v_{2}^{n}$ in bin $w_{2}$ of the second bin set,
such that $(v_{1}^{n},v_{2}^{n})\in
A_{\epsilon}^{(n)}(P_{V_{1},V_{2}})$ where
$A_{\epsilon}^{(n)}(P_{V_{1},V_{2}})$ denotes the set of jointly
typical sequences $v_{1}^{n}$ and $v_{2}^{n}$ with respect to
$P(v_{1},v_{2})$. The rates are such that there exist more than one
joint typical pair. The transmitter randomly chooses one of them and
then generates $x^{n}$ according to
$P(x^{n}|v_{1}^{n},v_{2}^{n})=\prod_{i=1}^{n}P(x_{i}|v_{1,i},v_{2,i})$.
This scheme is equivalent to the scenario in which each bin is
divided into subbins and the transmitter randomly chooses one of the
subbins of bin $w_{1}$ and one of the subbins of bin $w_{2}$. It
then looks for a joint typical sequence $(v_{1}^{n},v_{2}^{n})$ in
the corresponding subbins and generates $x^{n}$.

3) \textit{Decoding}: The received signals at the legitimate
receivers, $y_{1}^{n}$ and $y_{2}^{n}$, are the outputs of the
channels $P(y_{1}^{n}|x^{n})=\prod_{i=1}^{n}P(y_{1,i}|x_{i})$ and
$P(y_{2}^{n}|x^{n})=\prod_{i=1}^{n}P(y_{2,i}|x_{i})$, respectively.
The first receiver looks for the unique sequence $v_{1}^{n}$ such
that $(v_{1}^{n},y_{1}^{n})$ is jointly typical and declares the
index of the bin containing $v_{1}^{n}$ as the message received. The
second receiver uses the same method to extract the message $w_{2}$.

4) \textit{Error Probability Analysis}: Since the region of
(\ref{l0}) is a subset of Marton region, then the error probability
analysis is the same as \cite{3}.

5) \textit{Equivocation Calculation}: Please see Appendix A.
\end{proof}
The achievability scheme in Theorem \ref{t4} introduces random
binning. However, when we want to construct the rate region of
(\ref{eq14}), it is not clear how to choose the auxiliary random
variables $V_{1}$ and $V_{2}$. Here, we employ the Dirty-Paper
Coding (DPC) technique to develop the secret DPC (SDPC) achievable
rate region for the SAMBC. We consider a secret dirty-paper encoder
with Gaussian codebooks as follows.

First, we separate the channel input $\mathbf{x}$ into two random
vectors $\mathbf{b_{1}}$ and $\mathbf{b_{2}}$ such that
\begin{equation}
\mathbf{b_{1}}+\mathbf{b_{2}}=\mathbf{x}
\end{equation}
Here, $\mathbf{b_{1}}$ and $\mathbf{b_{2}}$ and $\mathbf{v_{1}}$ and
$\mathbf{v_{2}}$ are chosen as follows:
\begin{IEEEeqnarray}{lr}\nonumber
\mathbf{b_{1}}\sim \mathcal{N}(0,\mathbf{B_{1}}),\\ \nonumber
\mathbf{b_{2}}\sim \mathcal{N}(0,\mathbf{B_{2}}),\\
\nonumber \mathbf{v_{2}}=\mathbf{b_{2}},\\ \label{eq15}
\mathbf{v_{1}}=\mathbf{b_{1}}+\mathbf{C}\mathbf{b_{2}}.
\end{IEEEeqnarray}
where $\mathbf{B_{1}}=E[\mathbf{b_{1}}\mathbf{b_{1}^{T}}]\succeq0$
and $\mathbf{B_{2}}=E[\mathbf{b_{2}}\mathbf{b_{2}^{T}}]\succeq0$ are
covariance matrices such that $\mathbf{B_{1}}+\mathbf{B_{2}}\preceq
\mathbf{S}$, and the matrix $\mathbf{C}$ is given as follows:
\begin{equation}
\mathbf{C}=\mathbf{B_{1}}\left(\mathbf{N_{1}}+\mathbf{B_{1}}\right)^{-1}
\end{equation}
By substituting (\ref{eq15}) into the Theorem \ref{t4}, we obtain
the following SDPC rate region for the SAMBC.
\begin{lem}(SDPC Rate Region):
Let $\mathbf{S}$ be a positive semi-definite matrix. Then the
following SDPC rate region
 of an SAMBC with
a covariance matrix constraint $\mathbf{S}$ is achievable.
\begin{equation}
\mathcal{R}^{SDPC}(\mathbf{S},\mathbf{N_{1,2,3}})=conv\left\{\bigcup_{\pi\in\prod}\mathcal{R}^{SDPC}(\pi,\mathbf{S},\mathbf{N_{1,2,3}})\right\}
\end{equation}
where $\prod$ is the collection of all possible permutations of the
ordered set $\{1,2\}$, $conv$ is the convex closure operator and
$\mathcal{R}^{SDPC}(\pi,\mathbf{S},\mathbf{N_{1,2,3}})$ is given as
follows:
\begin{IEEEeqnarray}{rl}\nonumber
\mathcal{R}^{SDPC}(\pi,\mathbf{S},\mathbf{N_{1,2,3}})=   \left\{
                                                       \begin{array}{ll}
                                                          \left(R_{1},R_{2}\right)\big|  R_{k}=R_{\pi^{-1}(k)}^{SDPC}(\pi,\mathbf{B_{1,2}},\mathbf{N_{1,2,3}})~ k=1,2\\ \hbox{s.t}~~\mathbf{S}-(\mathbf{B_{1}}+\mathbf{B_{2}})\succeq 0,~\mathbf{B_{1}}\succeq 0, ~\mathbf{B_{2}}\succeq 0
                                                       \end{array}
                                                     \right\}.
\end{IEEEeqnarray}
where
\begin{IEEEeqnarray}{rl}\nonumber
R_{\pi^{-1}(k)}^{SDPC}\left(\pi,\mathbf{B_{1,2}},\mathbf{N_{1,2,3}}\right)=\frac{1}{2}\left[\log\frac{\left|\sum_{i=1}^{\pi^{-1}(k)}\mathbf{B_{\pi(i)}}+\mathbf{N_{k}}\right|}{\left|\sum_{i=1}^{\pi^{-1}(k)-1}\mathbf{B_{\pi(i)}}+\mathbf{N_{k}}\right|}-\frac{1}{2}\log\frac{\left|\sum_{i=1}^{\pi^{-1}(k)}\mathbf{B_{\pi(i)}}+\mathbf{N_{3}}\right|}{\left|\sum_{i=1}^{\pi^{-1}(k)-1}\mathbf{B_{\pi(i)}}+\mathbf{N_{3}}\right|}\right]^{+}
\end{IEEEeqnarray}
\end{lem}
Note that for the identity permutation, $\pi_{I}$, where
$\pi_{I}(k)=k$ we have,
\begin{IEEEeqnarray}{rl}\nonumber
\mathcal{R}^{SDPC}(\pi_{I},\mathbf{S},\mathbf{N_{1,2,3}})=\mathcal{R}^{G}(\mathbf{S},\mathbf{N_{1,2,3}})
\end{IEEEeqnarray}
\begin{proof}
We prove the lemma for the case of identity permutation
$\pi_{I}=\{1,2\}$. This proof can similarly be used for the case
that $\pi=\{2,1\}$. According to the Theorem \ref{t4}, we have,
\begin{IEEEeqnarray}{rl}\nonumber
    R_{1}&\leq \min\left\{I(V_{1};Y_{1})-I(V_{1};Z),I(V_{1};Y_{1})+I(V_{2};Z)-I(V_{1},V_{2};Z)-I(V_{1};V_{2})\right\}, \\ \nonumber
         &\stackrel{(a)}{\leq} \min\left\{I(V_{1};Y_{1})-I(V_{1};Z),I(V_{1};Y_{1})-I(V_{1};Z|V_{2})-I(V_{1};V_{2})\right\}, \\
\nonumber
         &\stackrel{(b)}{\leq} I(V_{1};Y_{1})-I(V_{1};Z|V_{2})-I(V_{1};V_{2}), \\ \label{eq16}
    R_{2}&\leq I(V_{2};Y_{2})-I(V_{2};Z),
  \end{IEEEeqnarray}
where $(a)$ follows from the fact that
$I(V_{1},V_{2};Z)=I(V_{2};Z)+I(V_{1};Z|V_{2})$ and $(b)$ follows
from the fact that
$I(V_{1};Z|V_{2})+I(V_{1};V_{2})=I(Z,V_{2};V_{1})\geq I(Z;V_{1})$.
To calculate the upper-bound of $R_{1}$, we need to review the
following lemma which has been noted by several authors \cite{12}.
\begin{lem}
Let $\mathbf{y_{1}}=\mathbf{b_{1}}+\mathbf{b_{2}}+\mathbf{n_{1}}$,
where $\mathbf{b_{1}}$, $\mathbf{b_{2}}$ and $\mathbf{n_{1}}$ are
Gaussian random vectors with covariance matrices $\mathbf{B_{1}}$,
$\mathbf{B_{2}}$ and $\mathbf{N_{1}}$ respectively. Let
$\mathbf{b_{1}}$, $\mathbf{b_{2}}$ and $\mathbf{n_{1}}$ be
independent, and let
$\mathbf{v_{1}}=\mathbf{b_{1}}+\mathbf{C}\mathbf{b_{2}}$, where
$\mathbf{C}$ is an $t\times t$ matrix. Then an optimal matrix
$\mathbf{C}$ which maximizes $I(\mathbf{v_{1}};\mathbf{y_{1}})-
I(\mathbf{v_{1}};\mathbf{b_{2}})$ is
$\mathbf{C}=\mathbf{B_{1}}\left(\mathbf{N_{1}}+\mathbf{B_{1}}\right)^{-1}$.
Further, the maximum value of $I(\mathbf{v_{1}};\mathbf{y_{1}}-
I(\mathbf{v_{1}};\mathbf{b_{2}})$ is
$I(\mathbf{v_{1}};\mathbf{y_{1}}|\mathbf{b_{2}})$.
\end{lem}
Now, using the
above Lemma and substituting (\ref{eq15}) into (\ref{eq16}), we
obtain the following achievable rate region when $\pi=\pi_{I}$.
\begin{IEEEeqnarray}{rl}\nonumber
    R_{1}&\leq\frac{1}{2}\left[\log\left|\mathbf{N_{1}^{-1}}(\mathbf{B_{1}}+\mathbf{N_{1}})\right|-\frac{1}{2}\log\left|\mathbf{N_{3}^{-1}}(\mathbf{B_{1}}+\mathbf{N_{3}})\right|\right]^{+},\\ \nonumber
   R_{2}&\leq\frac{1}{2}\left[\log\frac{\left|\mathbf{B_{1}}+\mathbf{B_{2}}+\mathbf{N_{2}}\right|}{\left|\mathbf{B_{1}}+\mathbf{N_{2}}\right|}-\frac{1}{2}\log\frac{\left|\mathbf{B_{1}}+\mathbf{B_{2}}+\mathbf{N_{3}}\right|}{\left|\mathbf{B_{1}}+\mathbf{N_{3}}\right|}\right]^{+}.
\end{IEEEeqnarray}
\end{proof}
\subsection{SAMBC- Converse Proof}
For the converse part, note that not all points on the boundary of
$\mathcal{R}^{SDPC}(\mathbf{S},\mathbf{N_{1,2,3}})$ can be directly
obtained using a single SDPC scheme. Instead, we must use
time-sharing between points corresponding to different permutations.
Therefore, unlike the SADBC case, we cannot use a similar notion to
the optimal Gaussian rate vectors, as not all the boundary points
can immediately characterized as a solution of an optimization
problem. Instead, as the SDPC region is convex by definition, we use
the notion of supporting hyperplanes of \cite{7} to define this
region.

In this section, we first define the supporting hyperplane of a
closed and bounded set. Then, we present the relation between the
ideas of a supporting hyperplane and the enhanced channel in
Theorem \ref{t5} This theorem is an extension of Theorem \ref{t2} to the
SAMBC case. Finally, we use Theorem \ref{t5} to prove that
$\mathcal{R}^{SDPC}(\mathbf{S},\mathbf{N_{1,2,3}})$ is indeed the
capacity region of the SAMBC.
\begin{defi}
The set
$\{\overline{R}=(R_{1},R_{2})|\gamma_{1}R_{1}+\gamma_{2}R_{2}=b\}$,
for fixed and given scalars $\gamma_{1},\gamma_{2}$ and, $b$, is a
supporting hyperplane of a closed and bounded set
$\mathcal{X}\subset \mathbb{R}^{m}$, if
$\gamma_{1}R_{1}+\gamma_{2}R_{2}\leq b$ $\forall (R_{1},R_{2})\in
\mathcal{X}$, with equality for at least one rate vector
$(R_{1},R_{2})\in \mathcal{X}$.
\end{defi}
Note that as $\mathcal{X}$ is closed and bounded,
$\max_{(R_{1},R_{2})\in
\mathcal{X}}\gamma_{1}R_{1}+\gamma_{2}R_{2}$, exists for any
$\gamma_{1},\gamma_{2}$. Thus, we always can find a supporting
hyperplane for the set $\mathcal{X}$. As
$\mathcal{R}^{SDPC}(\mathbf{S},\mathbf{N_{1,2,3}})$ is a closed and
convex set, for each rate pair of
$\overline{R}^{o}=(R_{1}^{o},R_{2}^{o})\notin\mathcal{R}^{SDPC}(\mathbf{S},\mathbf{N_{1,2,3}})$
which lies outside the set, there exists a separating hyperplane
$\{(R_{1},R_{2})|\gamma_{1}R_{1}+\gamma_{2}R_{2}=b\}$ where
$\gamma_{1}\geq0,\gamma_{1}\geq0,b\geq0$ and
\begin{IEEEeqnarray}{rl}\nonumber
\gamma_{1}R_{1}+\gamma_{2}R_{2}&\leq b~~~ \forall
(R_{1},R_{2})\in\mathcal{R}^{SDPC}(\mathbf{S},\mathbf{N_{1,2,3}})\\
\nonumber \gamma_{1}R_{1}^{o}+\gamma_{2}R_{2}^{o}&> b
\end{IEEEeqnarray}
The following theorem illustrates the relation between the ideas of
enhanced channel and a supporting hyperplane.
\begin{thm}\label{t5}
Consider a SAMBC with noise covariance matrices
$(\mathbf{N_{1}},\mathbf{N_{2}},\mathbf{N_{3}})$ and an average
transmit covariance matrix constraint $\mathbf{S}\succ 0$. Assume
that $\{(R_{1},R_{2})|\gamma_{1}R_{1}+\gamma_{2}R_{2}=b\}$ is a
supporting hyperplane of the rate region
$\mathcal{R}^{SDPC}(\pi_{I},\mathbf{S},\mathbf{N_{1,2,3}})$ such
that $0\leq \gamma_{1}\leq \gamma_{2}$, $\gamma_{2}>0$ and $b\geq0$.
Then, there exists an enhanced SADBC with noise covariance matrices
$(\mathbf{N_{1}^{'}},\mathbf{N_{2}^{'}},\mathbf{N_{3}^{'}})$ such
that the following properties hold.
\begin{enumerate}
  \item Enhancement:\\ \nonumber
   $\mathbf{N_{1}^{'}}\preceq \mathbf{N_{1}},~~~\mathbf{N_{2}^{'}}\preceq \mathbf{N_{2}},~~~\mathbf{N_{3}^{'}}= \mathbf{N_{3}},~~~\mathbf{N_{1}^{'}}\preceq \mathbf{N_{2}^{'}}$,
  \item Supporting hyperplane preservation: \\ \nonumber
  $\{(R_{1},R_{2})|\gamma_{1}R_{1}+\gamma_{2}R_{2}=b\}$ is also a
  supporting hyperplane of the rate region $\mathcal{R}^{G}(\mathbf{S},\mathbf{N_{1,2,3}^{'}})$
\end{enumerate}
\end{thm}
\begin{proof}
To prove this theorem, we can follow the steps of the proof of
Theorem \ref{t2}. Assume that the hyperplane
$\{(R_{1},R_{2})|\gamma_{1}R_{1}+\gamma_{2}R_{2}=b\}$ touches the
region $\mathcal{R}^{SDPC}(\pi_{I},\mathbf{S},\mathbf{N_{1,2,3}})$
at the pint $(R_{1}^{*},R_{2}^{*})$. Let $\mathbf{B_{1}^{*}},
\mathbf{B_{2}^{*}}$ be two positive semi-definite matrices such that
$\mathbf{B_{1}^{*}}+\mathbf{B_{2}^{*}}\preceq \mathbf{S}$ and such
that
\begin{IEEEeqnarray}{rl}\nonumber
R_{k}^{SDPC}(\pi_{I},\mathbf{B_{1,2}^{*}},\mathbf{N_{1,2,3}})=R_{k}^{*},~~~~~k=1,2
\end{IEEEeqnarray}
By definition of the supporting hyperplane, the scalar $b$ and the
matrices $(\mathbf{B_{1}^{*}},\mathbf{B_{2}^{*}})$ are the solution
of the following optimization problem:
\begin{IEEEeqnarray}{rl}\nonumber
&\max_{\mathbf{B_{1}},\mathbf{B_{2}}}\gamma_{1}R_{1}^{SDPC}(\pi_{I},\mathbf{B_{1,2}},\mathbf{N_{1,2,3}})+\gamma_{2}R_{2}^{SDPC}(\pi_{I},\mathbf{B_{1,2}},\mathbf{N_{1,2,3}})\\
\nonumber
&\hbox{s.t}~~~\mathbf{B_{1}}+\mathbf{B_{2}}\preceq\mathbf{S}~~~\mathbf{B_{k}}\succeq0~~~~~k=1,2
\end{IEEEeqnarray}
We define the noise covariance matrices of the enhanced SADBC as
(\ref{en}). Since for the permutation $\pi=\pi_{I}$ we have
$\mathcal{R}^{SDPC}(\pi_{I},\mathbf{S},\mathbf{N_{1,2,3}})=\mathcal{R}^{G}(\mathbf{S},\mathbf{N_{1,2,3}})$,
the supporting hyperplane
$\{(R_{1},R_{2})|\gamma_{1}R_{1}+\gamma_{2}R_{2}=b\}$ is also a
supporting hyperplane of the rate region
$\mathcal{R}^{G}(\mathbf{S},\mathbf{N_{1,2,3}^{'}})$.
\end{proof}
We can now use Theorem \ref{t5} and the capacity result of the SADBC
to prove that $\mathcal{R}^{SDPC}(\mathbf{S},\mathbf{N_{1,2,3}})$ is
indeed the capacity region of the SAMBC. The following theorem
formally states the main result of this section.
\begin{thm}\label{t6}
Consider a SAMBC with positive definite noise covariance matrices
$(\mathbf{N_{1}},\mathbf{N_{2}},\mathbf{N_{3}})$. Let
$\mathcal{C}(\mathbf{S},\mathbf{N_{1,2,3}})$ denote the capacity
region of the SAMBC under a covariance matrix constraint
$\mathbf{S}\succ 0$ .Then,
$\mathcal{C}(\mathbf{S},\mathbf{N_{1,2,3}})=\mathcal{R}^{SDPC}(\mathbf{S},\mathbf{N_{1,2,3}})$.
\end{thm}
\begin{proof}
To proof this theorem, we use Theorem \ref{t5} to show that for
every rate vector $\overline{R}^{o}$, which lies outside the region
$\mathcal{R}^{SDPC}(\mathbf{S},\mathbf{N_{1,2,3}})$, we can find an
enhanced SADBC, whose capacity region does not contain
$\overline{R}^{o}$. As the capacity region of the enhanced channel
outer bounds that of the original channel, therefore,
$\overline{R}^{o}$ cannot be an achievable rate vector.

Let $\overline{R}^{o}=(R_{1}^{o},R_{2}^{o})$ be a rate vector which
lies outside the region
$\mathcal{R}^{SDPC}(\mathbf{S},\mathbf{N_{1,2,3}})$. There exists a
supporting and separating hyperplane
$\{(R_{1},R_{2})|\gamma_{1}R_{1}+\gamma_{2}R_{2}=b\}$ where
$\gamma_{1}\geq 0$, $\gamma_{2}\geq0$, and at least one of the
$\gamma_{k}$'s is positive. Without loose of generality, we assume
that $\gamma_{2}\geq \gamma_{1}$. If that is not the case, we can
always reorder the indices of the users such that this assumption
will hold. By definition of the region
$\mathcal{R}^{SDPC}(\mathbf{S},\mathbf{N_{1,2,3}})$, we have,
\begin{IEEEeqnarray}{rl}\nonumber
\mathcal{R}^{SDPC}(\pi_{I},\mathbf{S},\mathbf{N_{1,2,3}})\subseteq
\mathcal{R}^{SDPC}(\mathbf{S},\mathbf{N_{1,2,3}}).
\end{IEEEeqnarray}
Note that, as $\{(R_{1},R_{2})|\gamma_{1}R_{1}+\gamma_{2}R_{2}=b\}$
is a supporting hyperplane of
$\mathcal{R}^{SDPC}(\mathbf{S},\mathbf{N_{1,2,3}})$, we can wire,
\begin{IEEEeqnarray}{rl}\nonumber
b^{'}&=\max_{(R_{1},R_{2})\in\mathcal{R}^{SDPC}(\pi_{I},\mathbf{S},\mathbf{N_{1,2,3}})}\gamma_{1}R_{1}+\gamma_{2}R_{2}\\
\nonumber &\leq \max_{(R_{1},R_{2})\in
\mathcal{R}^{SDPC}(\mathbf{S},\mathbf{N_{1,2,3}})}\gamma_{1}R_{1}+\gamma_{2}R_{2}=b.
\end{IEEEeqnarray}
Furthermore, we can also write,
\begin{IEEEeqnarray}{rl}\nonumber
\gamma_{1}R_{1}^{o}+\gamma_{2}R_{2}^{o}>b\geq b^{'}.
\end{IEEEeqnarray}
Therefore, the hyperplane of $\{(R_{1},R_{2})|\gamma_{1}R_{1}+\gamma_{2}R_{2}=b^{'}\}$
is a supporting and separating hyperplane for the rate region
$\mathcal{R}^{SDPC}(\pi_{I},\mathbf{S},\mathbf{N_{1,2,3}})$. By
Theorem \ref{t5}, we know that there exists an enhanced SADBC whose
Gaussian rate region
$\mathcal{R}^{G}(\mathbf{S},\mathbf{N_{1,2,3}^{'}})$ lies under the
supporting hyperplane and hence
$(R_{1}^{o},R_{2}^{o})\notin\mathcal{R}^{G}(\mathbf{S},\mathbf{N_{1,2,3}^{'}})$.
Therefore, $(R_{1}^{o},R_{2}^{o})$ must lies outside the capacity
region of the enhanced SADBC. To complete the proof, note that the
capacity region of the enhanced SADBC contains that of the original
channel and therefore, $(R_{1}^{o},R_{2}^{o})$ must lies outside the
capacity region of the original SAMBC. As this statement is true for
all rate vectors which lie outside
$\mathcal{R}^{SDPC}(\mathbf{S},\mathbf{N_{1,2,3}})$, therefore we
have
$\mathcal{C}(\mathbf{S},\mathbf{N_{1,2,3}})\subseteq\mathcal{R}^{SDPC}(\mathbf{S},\mathbf{N_{1,2,3}})$.
However, $\mathcal{R}^{SDPC}(\mathbf{S},\mathbf{N_{1,2,3}})$ is the
set of achievable rates and therefore,
$\mathcal{C}(\mathbf{S},\mathbf{N_{1,2,3}})=\mathcal{R}^{SDPC}(\mathbf{S},\mathbf{N_{1,2,3}})$.
\end{proof}
With the same discussion of \cite{7}, the result of SAMBC can extend
to the SGMBC and may be omitted here. The results of the secrecy
capacity region for two receiver can be extended for $m$ receivers
as follows.

\begin{coro}\label{c1}
Consider a SGMBC with $m$ receivers and one external eavesdropper.
Let $\mathbf{S}$ be a positive semi-definite matrix. Then the SDPC
rate region of
$\mathcal{R}^{SDPC}(\mathbf{S},\mathbf{N_{1,...,m},\mathbf{H_{1,...,m}}})$,
which is defined by the following convex closure is indeed the
secrecy capacity region of the SGMBC under a covariance constraint
$\mathbf{S}$.
\begin{equation}
\mathcal{R}^{SDPC}(\mathbf{S},\mathbf{N_{1,...,m}},\mathbf{H_{1,...,m}})=conv\left\{\bigcup_{\pi\in\prod}\mathcal{R}^{SDPC}(\pi,\mathbf{S},\mathbf{N_{1,...,m}},\mathbf{H_{1,...,m}})\right\}
\end{equation}
where $\prod$ is the collection of all possible permutations of the
ordered set $\{1,...,m\}$, $conv$ is the convex closure operator and
$\mathcal{R}^{SDPC}(\pi,\mathbf{S},\mathbf{N_{1,...,m}},\mathbf{H_{1,...,m}})$
is given as follows:
\begin{IEEEeqnarray}{rl}\nonumber
\mathcal{R}^{SDPC}(\pi,\mathbf{S},\mathbf{N_{1,...,m}},\mathbf{H_{1,...,m}})=
\left\{
                                                       \begin{array}{ll}
                                                          \left(R_{1},R_{2}\right)\big|  R_{k}=R_{\pi^{-1}(k)}^{SDPC}(\pi,\mathbf{B_{1,...,m}},\mathbf{N_{1,...,m}},\mathbf{H_{1,...,m}})~ k=1,...,m\\ \hbox{s.t}~~\mathbf{S}-\sum_{i=1}^{m}\mathbf{B}_{i}\succeq 0,~\mathbf{B}_{i}\succeq
                                                          0,~i=1,...,m
                                                       \end{array}
                                                     \right\}.
\end{IEEEeqnarray}
where
\begin{IEEEeqnarray}{rl}\label{eq17}
R_{\pi^{-1}(k)}^{SDPC}\left(\pi,\mathbf{B_{1,...,m}},\mathbf{N_{1,...,m}},\mathbf{H_{1,...,m}}\right)&=\frac{1}{2}\big[\log\frac{\left|\mathbf{H_{k}}\left(\sum_{i=1}^{\pi^{-1}(k)}\mathbf{B_{\pi(i)}}\right)\mathbf{H_{k}^{\dag}}+\mathbf{N_{k}}\right|}{\left|\mathbf{H_{k}}\left(\sum_{i=1}^{\pi^{-1}(k)-1}\mathbf{B_{\pi(i)}}\right)\mathbf{H_{k}^{\dag}}+\mathbf{N_{k}}\right|}\\
\nonumber&-\frac{1}{2}\log\frac{\left|\mathbf{H_{3}}\left(\sum_{i=1}^{\pi^{-1}(k)}\mathbf{B_{\pi(i)}}\right)\mathbf{H_{3}^{\dag}}+\mathbf{N_{3}}\right|}{\left|\mathbf{H_{3}}\left(\sum_{i=1}^{\pi^{-1}(k)-1}\mathbf{B_{\pi(i)}}\right)\mathbf{H_{3}^{\dag}}+\mathbf{N_{3}}\right|}\big]^{+}
\end{IEEEeqnarray}

\end{coro}
\section{Multiple-Input Single-Outputs Multiple Eavesdropper (MISOME) Channel}
In this section we investigate practical characterizations for the
specific scenario in which the transmitter and the eavesdropper have
multiple antennas, while both intended receivers have a single
antenna. We refer to this configuration as the MISOME case. The
significance of this model is when a base station wishes to
broadcast secure information for small mobile units. In this
scenario small mobile units have single antenna while the base
station and the eavesdropper can afford multiple antennas. We can
rewrite the signals received by the destination and the eavesdropper
for the MISOME channel as follows.
\begin{IEEEeqnarray}{lr}\nonumber
y_{1}=\mathbf{h}_{1}^{\dag}\mathbf{x}+n_{1},\\\label{misome}
y_{2}=\mathbf{h}_{2}^\dag\mathbf{x}+n_{2},\\
\mathbf{z}=\mathbf{H}_{3}\mathbf{x}+\mathbf{n_{3}}, \nonumber
\end{IEEEeqnarray}
where $\mathbf{h}_{1}$ and $\mathbf{h}_{2}$ are fixed, real gain
matrices which model the channel gains between the transmitter and
the legitimate receivers. These are matrices of size $t \times 1$.
The channel state information again is assumed to be known perfectly
at the transmitter and at all receivers. Here, the superscript
$\dag$ denotes the Hermitian transpose of a vector. Without lost of
generality, we assume that $n_{1}$ and $n_{2}$ are i.i.d real
Gaussian random variables with zero means unit covariances, i.e.,
$n_{1},n_{2}\sim\mathcal{N}(0,1)$. Furthermore, we assume that
$\mathbf{n_{3}}$ is a Gaussian random vector with zero mean and
covariance matrix $\mathbf{I}$. In this section, we assume that the
input $\mathbf{x}$ satisfies a total power constraint of $P$, i.e.,
\begin{equation}\nonumber
Tr\{E(\mathbf{x}\mathbf{x}^{T})\}\leq P
\end{equation}
Before we state our results for the MISOME channel, we need to
review some properties of generalized eigenvalues and eigenvectors.
For more details of this topic, see, e.g.,\cite{13}.
\begin{defi}(Generalized eigenvalue-eigenvector)
Let $\mathbf{A}$ be a Hermitian matrix and $\mathbf{B}$ be a
positive definite matrix. Then, $(\lambda,\mbox{\boldmath$\psi$})$
is a generalized eigenvalue-eigenvector pair if it satisfy the
following equation.
\begin{equation}\nonumber
\mathbf{A}\mbox{\boldmath$\psi$}=\lambda
\mathbf{B}\mbox{\boldmath$\psi$}
\end{equation}
\end{defi}
Note that as $\mathbf{B}$ is invertible, the generalized eigenvalues
and eigenvectors of the pair $(\mathbf{A},\mathbf{B})$ are the
regular eigenvalues and eigenvectors of the matrix
$\mathbf{B}^{-1}\mathbf{A}$. The following Lemma, describes the
variational characterization of the generalized
eigenvalue-eigenvector pair.
\begin{lem}(Variational Characterization)
Let $r(\mbox{\boldmath$\psi$})$ be the Rayleigh quotient defined as
the following.
\begin{equation}\nonumber
r(\mbox{\boldmath$\psi$})=\frac{\mbox{\boldmath$\psi^{\dag}$}\mathbf{A}\mbox{\boldmath$\psi$}}{\mbox{\boldmath$\psi^{\dag}$}\mathbf{B}\mbox{\boldmath$\psi$}}
\end{equation}
Then, the generalized eigenvectors of $(\mathbf{A},\mathbf{B})$ are
the stationary point solution of the Rayleigh quotient
$r(\mbox{\boldmath$\psi$})$. Specifically, the largest generalized
eigenvalue $\lambda_{\max}$ is the maximum of the Rayleigh quotient
$r(\mbox{\boldmath$\psi$})$ and the optimum is attained by the
eigenvector $\mbox{\boldmath$\psi$}_{\max}$ which is corresponded to
$\lambda_{\max}$, i.e.,
\begin{equation}\nonumber
\max_{\mbox{\boldmath$\psi$}}r(\mbox{\boldmath$\psi$})=\frac{\mbox{\boldmath$\psi^{\dag}$}_{\max}\mathbf{A}\mbox{\boldmath$\psi$}_{\max}}{\mbox{\boldmath$\psi^{\dag}$}_{\max}\mathbf{B}\mbox{\boldmath$\psi$}_{\max}}=\lambda_{\max}
\end{equation}
\end{lem}
Now consider the MISOME channel of (\ref{misome}). Assume that
$0\leq\alpha\leq1$ and $P$ are fixed. Let define the following
matrices for this channel.
\begin{IEEEeqnarray}{rl}\nonumber
\mathbf{A_{1,1}}&=\mathbf{I}+\alpha
P\mathbf{h_{1}}\mathbf{h_{1}^{\dag}},\\ \nonumber
\mathbf{B_{1,1}}&=\mathbf{I}+\alpha
P\mathbf{H_{3}^{\dag}}\mathbf{H_{3}}
\end{IEEEeqnarray}
Suppose that $(\lambda_{(1,1)\max},\mbox{\boldmath$\psi$}_{1\max})$
is the largest generalized eigenvalue and the corresponding
eigenvector pair of the pencil
$(\mathbf{A_{1,1}},\mathbf{B_{1,1}})$. We furthermore define the
following matrices for the MISOME channel.
\begin{IEEEeqnarray}{rl}\nonumber
\mathbf{A_{2,2}}&=\mathbf{I}+\frac{(1-\alpha)P}{1+\alpha P|\mathbf{h_{2}^{\dag}}\mbox{\boldmath$\psi$}_{1\max}|^{2}}\mathbf{h_{2}}\mathbf{h_{2}^{\dag}},\\
\nonumber
\mathbf{B_{2,2}}&=\mathbf{I}+(1-\alpha)P\mathbf{H_{3}^{\dag}}\left(\mathbf{I}+\alpha
P\mathbf{H_{3}}\mbox{\boldmath$\psi$}_{1\max}\mbox{\boldmath$\psi$}_{1\max}^{\dag}\mathbf{H_{3}}\right)^{-1}\mathbf{H_{3}}
\end{IEEEeqnarray}
Assume that $(\lambda_{(2,2)\max},\mbox{\boldmath$\psi$}_{2\max})$
is the largest generalized eigenvalue and the corresponding
eigenvector pair of the pencil
$(\mathbf{A_{2,2}},\mathbf{B_{2,2}})$. Moreover, consider the
following matrices for this channel.
\begin{IEEEeqnarray}{rl}\nonumber
\mathbf{A_{2,1}}&=\mathbf{I}+(1-\alpha)
P\mathbf{h_{2}}\mathbf{h_{2}^{\dag}},\\ \nonumber
\mathbf{B_{2,1}}&=\mathbf{I}+(1-\alpha)
P\mathbf{H_{3}^{\dag}}\mathbf{H_{3}},\\
\nonumber\mathbf{A_{1,2}}&=\mathbf{I}+\frac{\alpha P}{1+(1-\alpha) P|\mathbf{h_{1}^{\dag}}\mbox{\boldmath$\psi$}_{3\max}|^{2}}\mathbf{h_{1}}\mathbf{h_{1}^{\dag}},\\
\nonumber \mathbf{B_{1,2}}&=\mathbf{I}+\alpha
P\mathbf{H_{3}^{\dag}}\left(\mathbf{I}+(1-\alpha)
P\mathbf{H_{3}}\mbox{\boldmath$\psi$}_{3\max}\mbox{\boldmath$\psi$}_{3\max}^{\dag}\mathbf{H_{3}}\right)^{-1}\mathbf{H_{3}},
\end{IEEEeqnarray}
where we assume that
$(\lambda_{(2,1)\max},\mbox{\boldmath$\psi$}_{3\max})$, and
$(\lambda_{(1,2)\max},\mbox{\boldmath$\psi$}_{4\max})$ are the
largest generalized eigenvalue and the corresponding eigenvector
pair of the pencils $(\mathbf{A_{2,1}},\mathbf{B_{2,1}})$, and
$(\mathbf{A_{1,2}},\mathbf{B_{1,2}})$ respectively. The following
theorem then characterizes the capacity region of the MISOME channel
under a total power constraint $P$ based on the above parameters.
\begin{thm}
Let $\mathcal{C}^{MISOME}$ denote the secrecy capacity region of the
the MISOME channel under an average total power constraint $P$. Let
$\prod$ be the collection of all possible permutations of the
ordered set $\{1,2\}$ and $conv$ be the convex closure operator,
then $\mathcal{C}^{MISOME}$ is given as follows.
\begin{equation}\nonumber
\mathcal{C}^{MISOME}=conv\left\{\bigcup_{\pi\in\prod}\mathcal{R}^{MISOME}(\pi)\right\}
\end{equation}
 where $\mathcal{R}^{MISOME}(\pi)$ is given as follows.
\begin{IEEEeqnarray}{rl}\nonumber
\mathcal{R}^{MISOME}(\pi)=\bigcup_{0\leq \alpha \leq
1}\mathcal{R}^{MISOME}(\pi,\alpha)
\end{IEEEeqnarray}
where $\mathcal{R}^{MISOME}(\pi,\alpha)$ is the set of all
$(R_{1},R_{2})$ satisfying the following condition.
\begin{IEEEeqnarray}{rl}\nonumber
    R_{k}&\leq \frac{1}{2}\left[\log\lambda_{(k,\pi^{-1}(k))\max}\right]^{+}, ~~~~~~~~
    k=1,2.
\end{IEEEeqnarray}
  \end{thm}
\begin{proof}
This theorem is a special case of Theorem \ref{t6} and corollary
\ref{c1}. First assume that the permutation $\pi=\pi_{I}=\{1,2\}$.
In the SDPC achievable rate region of (\ref{eq17}), we choose the
covariance matrices $\mathbf{B_{1}}$ and $\mathbf{B_{2}}$ are as
follows.
\begin{IEEEeqnarray}{rl}\nonumber
\mathbf{B_{1}}&=\alpha
P\mbox{\boldmath$\psi$}_{1\max}\mbox{\boldmath$\psi$}_{1\max}^{\dag},\\
\nonumber \mathbf{B_{2}}&=(1-\alpha)
P\mbox{\boldmath$\psi$}_{2\max}\mbox{\boldmath$\psi$}_{2\max}^{\dag}.
\end{IEEEeqnarray}
In other words, the channel input $\mathbf{x}$ is separated into two
vectors $\mathbf{b_{1}}$ and $\mathbf{b_{2}}$ such that
\begin{IEEEeqnarray}{rl}\nonumber
\mathbf{x}&=\mathbf{b_{1}}+\mathbf{b_{2}},\\ \nonumber
\mathbf{b_{1}}&=u_{1}\mbox{\boldmath$\psi$}_{1\max},\\
\nonumber \mathbf{b_{2}}&=u_{2}\mbox{\boldmath$\psi$}_{2\max}.
\end{IEEEeqnarray}
where $u_{1}\sim\mathcal{N}(0,\alpha P)$,
$u_{2}\sim\mathcal{N}(0,(1-\alpha) P)$, and $0\leq \alpha \leq 1$.
Using these parameters, the region of
$\mathcal{R}^{SDPC}(\pi_{I},\mathbf{S},\mathbf{N_{1,2,3}})$ becomes
as follows.
\begin{IEEEeqnarray}{rl}\nonumber
R_{1}&\leq\frac{1}{2}\left[\log\left|1+\mathbf{h_{1}^{\dag}}\mathbf{B_{1}}\mathbf{h_{1}}\right|-\frac{1}{2}\log\left|\mathbf{I}+\mathbf{H_{3}}\mathbf{B_{1}}\mathbf{H_{3}^{\dag}}\right|\right]^{+},\\
\nonumber &=\frac{1}{2}\left[\log\left|(1+\alpha
P\mathbf{h_{1}^{\dag}}\mbox{\boldmath$\psi$}_{1\max}\mbox{\boldmath$\psi$}_{1\max}^{\dag}\mathbf{h_{1}})\right|-\frac{1}{2}\log\left|(\mathbf{I}+\alpha
P\mathbf{H_{3}}\mbox{\boldmath$\psi$}_{1\max}\mbox{\boldmath$\psi$}_{1\max}^{\dag}\mathbf{H_{3}^{\dag}})\right|\right]^{+},\\
\nonumber
&=\frac{1}{2}\left[\log\frac{\mbox{\boldmath$\psi$}_{1\max}^{\dag}\left(\mathbf{I}+\alpha
P\mathbf{h_{1}}\mathbf{h_{1}^{\dag}}\right)\mbox{\boldmath$\psi$}_{1\max}}{\mbox{\boldmath$\psi$}_{1\max}^{\dag}\left(\mathbf{I}+\alpha
P\mathbf{H_{3}^{\dag}}\mathbf{H_{3}}\right)\mbox{\boldmath$\psi$}_{1\max}}\right]^{+},\\
\label{hsnr1}
&\stackrel{(a)}{=}\frac{1}{2}\left[\log\lambda_{(1,1)\max}\right]^{+}
\end{IEEEeqnarray}
where $(a)$ is due to the fact that
$|\mathbf{I}+\mathbf{A}\mathbf{B}|=|\mathbf{I}+\mathbf{B}\mathbf{A}|$
and the fact that
$\mbox{\boldmath$\psi$}_{1\max}^{\dag}\mbox{\boldmath$\psi$}_{1\max}=1$
Similarly for the $R_{2}$ we have,
\begin{IEEEeqnarray}{rl}\nonumber
R_{2}&\leq\frac{1}{2}\left[\log
\frac{\left|1+\mathbf{h_{2}^{\dag}}\left(\mathbf{B_{1}}+\mathbf{B_{2}}\right)\mathbf{h_{2}}\right|}{\left|1+\mathbf{h_{2}^{\dag}}\mathbf{B_{1}}\mathbf{h_{2}}\right|}-\frac{1}{2}\log\frac{\left|\mathbf{I}+\mathbf{H_{3}}\left(\mathbf{B_{1}}+\mathbf{B_{2}}\right)\mathbf{H_{3}^{\dag}}\right|}{\left|\mathbf{I}+\mathbf{H_{3}}\mathbf{B_{1}}\mathbf{H_{3}^{\dag}}\right|}\right]^{+}\\
\nonumber&=\frac{1}{2}\left[\log
\left|1+\frac{\mathbf{h_{2}^{\dag}}\mathbf{B_{2}}\mathbf{h_{2}}}{1+\mathbf{h_{2}^{\dag}}\mathbf{B_{1}}\mathbf{h_{2}}}\right|-\frac{1}{2}\log\left|\mathbf{I}+\frac{\mathbf{H_{3}}\mathbf{B_{2}}\mathbf{H_{3}^{\dag}}}{\mathbf{I}+\mathbf{H_{3}}\mathbf{B_{1}}\mathbf{H_{3}^{\dag}}}\right|\right]^{+}\\
\label{hsnr2} &=\frac{1}{2}\left[\log
\frac{\mbox{\boldmath$\psi$}_{2\max}^{\dag}\left(\mathbf{I}+\frac{(1-\alpha)
P\mathbf{h_{2}}\mathbf{h_{2}^{\dag}}}{1+\alpha
P|\mathbf{h_{2}^{\dag}}\mbox{\boldmath$\psi$}_{1\max}|^{2}}\right)\mbox{\boldmath$\psi$}_{2\max}}{\mbox{\boldmath$\psi$}_{2\max}^{\dag}\left(\mathbf{I}+(1-\alpha)P\mathbf{H_{3}^{\dag}}\left(\mathbf{I}+\alpha
P\mathbf{H_{3}}\mbox{\boldmath$\psi$}_{1\max}\mbox{\boldmath$\psi$}_{1\max}^{\dag}\mathbf{H_{3}^{\dag}}\right)^{-1}
\mathbf{H_{3}}\right)\mbox{\boldmath$\psi$}_{2\max}}\right]^{+}=\frac{1}{2}\left[\log\lambda_{(2,2)\max}\right]^{+}.
\end{IEEEeqnarray}
Similarly, when $\pi=\{2,1\}$, in the SDPC region, we choose
$\mathbf{b_{1}}=u_{1}\mbox{\boldmath$\psi$}_{4\max}$ and
$\mathbf{b_{2}}=u_{2}\mbox{\boldmath$\psi$}_{3\max}$. Then the SDPC
region is given as follows.
\begin{IEEEeqnarray}{rl}\nonumber
R_{1}&\leq\frac{1}{2}\left[\log
\frac{\left|1+\mathbf{h_{1}^{\dag}}\left(\mathbf{B_{1}}+\mathbf{B_{2}}\right)\mathbf{h_{1}}\right|}{\left|1+\mathbf{h_{1}^{\dag}}\mathbf{B_{2}}\mathbf{h_{1}}\right|}-\frac{1}{2}\log\frac{\left|\mathbf{I}+\mathbf{H_{3}}\left(\mathbf{B_{1}}+\mathbf{B_{2}}\right)\mathbf{H_{3}^{\dag}}\right|}{\left|\mathbf{I}+\mathbf{H_{3}}\mathbf{B_{2}}\mathbf{H_{3}^{\dag}}\right|}\right]^{+}\\
\nonumber&=\frac{1}{2}\left[\log
\left|1+\frac{\mathbf{h_{1}^{\dag}}\mathbf{B_{1}}\mathbf{h_{1}}}{1+\mathbf{h_{1}^{\dag}}\mathbf{B_{2}}\mathbf{h_{1}}}\right|-\frac{1}{2}\log\left|\mathbf{I}+\frac{\mathbf{H_{3}}\mathbf{B_{1}}\mathbf{H_{3}^{\dag}}}{\mathbf{I}+\mathbf{H_{3}}\mathbf{B_{2}}\mathbf{H_{3}^{\dag}}}\right|\right]^{+}\\
\nonumber&=\frac{1}{2}\left[\log
\frac{\mbox{\boldmath$\psi$}_{4\max}^{\dag}\left(\mathbf{I}+\frac{\alpha
P\mathbf{h_{1}}\mathbf{h_{1}^{\dag}}}{1+(1-\alpha)
P|\mathbf{h_{1}^{\dag}}\mbox{\boldmath$\psi$}_{3\max}|^{2}}\right)\mbox{\boldmath$\psi$}_{4\max}}{\mbox{\boldmath$\psi$}_{4\max}^{\dag}\left(\mathbf{I}+\alpha
P\mathbf{H_{3}^{\dag}}\left(\mathbf{I}+(1-\alpha)
P\mathbf{H_{3}}\mbox{\boldmath$\psi$}_{3\max}\mbox{\boldmath$\psi$}_{3\max}^{\dag}\mathbf{H_{3}^{\dag}}\right)^{-1}
\mathbf{H_{3}}\right)\mbox{\boldmath$\psi$}_{4\max}}\right]^{+}=\frac{1}{2}\left[\log\lambda_{(1,2)\max}\right]^{+},
\end{IEEEeqnarray}
and $R_{2}$ is bounded as follows
\begin{IEEEeqnarray}{rl}\nonumber
R_{2}&\leq\frac{1}{2}\left[\log\left|1+\mathbf{h_{2}^{\dag}}\mathbf{B_{2}}\mathbf{h_{2}}\right|-\frac{1}{2}\log\left|\mathbf{I}+\mathbf{H_{3}}\mathbf{B_{2}}\mathbf{H_{3}^{\dag}}\right|\right]^{+},\\
\nonumber &=\frac{1}{2}\left[\log\left|(1+(1-\alpha)
P\mathbf{h_{2}^{\dag}}\mbox{\boldmath$\psi$}_{3\max}\mbox{\boldmath$\psi$}_{3\max}^{\dag}\mathbf{h_{2}})\right|-\frac{1}{2}\log\left|(\mathbf{I}+(1-\alpha)
P\mathbf{H_{3}}\mbox{\boldmath$\psi$}_{3\max}\mbox{\boldmath$\psi$}_{3\max}^{\dag}\mathbf{H_{3}^{\dag}})\right|\right]^{+},\\
\nonumber
&=\frac{1}{2}\left[\log\frac{\mbox{\boldmath$\psi$}_{3\max}^{\dag}\left(\mathbf{I}+(1-\alpha)
P\mathbf{h_{2}}\mathbf{h_{2}^{\dag}}\right)\mbox{\boldmath$\psi$}_{3\max}}{\mbox{\boldmath$\psi$}_{3\max}^{\dag}\left(\mathbf{I}+(1-\alpha)
P\mathbf{H_{3}^{\dag}}\mathbf{H_{3}}\right)\mbox{\boldmath$\psi$}_{3\max}}\right]^{+},\\
\nonumber &=\frac{1}{2}\left[\log\lambda_{(2,1)\max}\right]^{+}.
\end{IEEEeqnarray}
\end{proof}
Note that the eigenvalues
$\lambda_{(l,k)\max}=\lambda_{(l,k)\max}(\alpha,P)$ and the
eigenvector
$\mbox{\boldmath$\psi$}_{k\max}=\mbox{\boldmath$\psi$}_{k\max}(\alpha,P),~~l,k=1,2$
are the functions of $\alpha$ and $P$. The following corollary
characterizes the secrecy capacity region of the MISOME channel in
high SNR regime.
\begin{coro}
In the high SNR regime, the secrecy capacity region of the MISOME
channel is given as follows.
\begin{equation}\nonumber
\lim_{P\rightarrow\infty}\mathcal{C}^{MISOME}=conv\left\{\bigcup_{\pi\in\prod}\mathcal{R}^{MISOME}_{\infty}(\pi)\right\}
\end{equation}
where
\begin{IEEEeqnarray}{rl}\nonumber
&\mathcal{R}^{MISOME}_{\infty}(\pi=\{1,2\})=\\
\nonumber&\left\{(R_{1},R_{2}), R_{1}\leq
\frac{1}{2}\left[\log\lambda_{\max}\left(\mathbf{h_{1}}\mathbf{h_{1}^{\dag}},\mathbf{H_{3}^{\dag}}\mathbf{H_{3}}\right)\right]^{+},R_{2}\leq
\frac{1}{2}\left[\log\frac{\lambda_{\max}\left(\mathbf{h_{2}}\mathbf{h_{2}^{\dag}},\mathbf{H_{3}^{\dag}}\mathbf{H_{3}}\right)}{b}\right]^{+}\right\}\\
\nonumber &\mathcal{R}^{MISOME}_{\infty}(\pi=\{2,1\})=\\
\nonumber&\left\{(R_{1},R_{2}), R_{1}\leq
\frac{1}{2}\left[\log\frac{\lambda_{\max}\left(\mathbf{h_{1}}\mathbf{h_{1}^{\dag}},\mathbf{H_{3}^{\dag}}\mathbf{H_{3}}\right)}{a}\right]^{+},R_{2}\leq
\frac{1}{2}\left[\log\lambda_{\max}\left(\mathbf{h_{2}}\mathbf{h_{2}^{\dag}},\mathbf{H_{3}^{\dag}}\mathbf{H_{3}}\right)\right]^{+}\right\}
\end{IEEEeqnarray}
where
$(\lambda_{\max}(\mathbf{A_{i}},\mathbf{B}),\mbox{\boldmath$\psi$}_{i\max})$
denotes the largest eigenvalue and corresponding eigenvector of the
pencil $(\mathbf{A_{i}},\mathbf{B})$ and
$b=\frac{|\mathbf{h_{2}^{\dag}}\mbox{\boldmath$\psi$}_{1\max}|^{2}}{\|\mathbf{H_{3}}\mbox{\boldmath$\psi$}_{1\max}\|^{2}}
$,
$a=\frac{|\mathbf{h_{1}^{\dag}}\mbox{\boldmath$\psi$}_{2\max}|^{2}}{\|\mathbf{H_{3}}\mbox{\boldmath$\psi$}_{2\max}\|^{2}}
$.
\end{coro}
Note that the above secrecy rate region is independent of $\alpha$
and therefore is a convex closure of two rectangular regions.
\begin{proof}
We restrict our attention to the case that
$\lambda_{(l,k)\max}(\alpha,P)>1$ for $l,k=1,2$ where the rates
$R_{1}$ and $R_{2}$ are nonzero. First suppose that
$\pi=\pi_{I}=\{1,2\}$. We show that
\begin{IEEEeqnarray}{rl}\label{tmp1}
\lim_{P\rightarrow\infty}\lambda_{(1,1)\max}(\alpha,P)&=\lambda_{\max}\left(\mathbf{h_{1}}\mathbf{h_{1}^{\dag}},\mathbf{H_{3}^{\dag}}\mathbf{H_{3}}\right)\\\label{tmp2}
\lim_{P\rightarrow\infty}\lambda_{(2,2)\max}(\alpha,P)&=\frac{\lambda_{\max}\left(\mathbf{h_{2}}\mathbf{h_{2}^{\dag}},\mathbf{H_{3}^{\dag}}\mathbf{H_{3}}\right)}{b}.
\end{IEEEeqnarray}
Note that since
\begin{equation}\nonumber
\lambda_{(1,1)\max}(\alpha,P)=\frac{1+\alpha P
|\mathbf{h_{1}^{\dag}}\mbox{\boldmath$\psi$}_{1\max}(\alpha,P)|^{2}}{1+\alpha
P \|\mathbf{H_{3}}\mbox{\boldmath$\psi$}_{1\max}(\alpha,P)\|^{2}}>1
\end{equation}
where
\begin{equation}\nonumber
\mbox{\boldmath$\psi$}_{1\max}(\alpha,P)=\arg
\max_{\left\{\mbox{\boldmath$\psi$}_{1}:\|\mbox{\boldmath$\psi$}_{1}\|^{2}=1\right\}}\frac{1+\alpha
P
|\mathbf{h_{1}^{\dag}}\mbox{\boldmath$\psi$}_{1}(\alpha,P)|^{2}}{1+\alpha
P \|\mathbf{H_{3}}\mbox{\boldmath$\psi$}_{1}(\alpha,P)\|^{2}}
\end{equation}
for all $P>0$ we have,
\begin{equation}\nonumber
|\mathbf{h_{1}^{\dag}}\mbox{\boldmath$\psi$}_{1\max}(\alpha,P)|^{2}>|\mathbf{H_{3}}\mbox{\boldmath$\psi$}_{1\max}(\alpha,P)\|^{2}
\end{equation}
Therefore, $\lambda_{(1,1)\max}$ is an increasing function of $P$.
Thus,
\begin{IEEEeqnarray}{rl}\nonumber
\lambda_{(1,1)\max}(\alpha,P)&\leq\frac{|\mathbf{h_{1}^{\dag}}\mbox{\boldmath$\psi$}_{1\max}(\alpha,P)|^{2}}{\|\mathbf{H_{3}}\mbox{\boldmath$\psi$}_{1\max}(\alpha,P)\|^{2}}\\
\nonumber&\leq\lambda_{\max}\left(\mathbf{h_{1}}\mathbf{h_{1}^{\dag}},\mathbf{H_{3}^{\dag}}\mathbf{H_{3}}\right)
\end{IEEEeqnarray}
Since
$\lambda_{\max}\left(\mathbf{h_{1}}\mathbf{h_{1}^{\dag}},\mathbf{H_{3}^{\dag}}\mathbf{H_{3}}\right)$
is independent of $P$ we have
\begin{IEEEeqnarray}{rl}\nonumber
\lim_{P\rightarrow\infty}\lambda_{(1,1)\max}\leq\lambda_{\max}\left(\mathbf{h_{1}}\mathbf{h_{1}^{\dag}},\mathbf{H_{3}^{\dag}}\mathbf{H_{3}}\right)
\end{IEEEeqnarray}
Next, defining
\begin{equation}\nonumber
\mbox{\boldmath$\psi$}_{1}(\infty)=\arg
\max_{\left\{\mbox{\boldmath$\psi$}_{1}:\|\mbox{\boldmath$\psi$}_{1}\|^{2}=1\right\}}\frac{
|\mathbf{h_{1}^{\dag}}\mbox{\boldmath$\psi$}_{1}|^{2}}{
\|\mathbf{H_{3}}\mbox{\boldmath$\psi$}_{1}\|^{2}}
\end{equation}
we have the following lower bound
\begin{IEEEeqnarray}{rl}\nonumber
\lim_{P\rightarrow\infty}\lambda_{(1,1)\max}(\alpha,P)&\geq\lim_{P\rightarrow\infty}\frac{\frac{1}{P}+\alpha|\mathbf{h_{1}^{\dag}}\mbox{\boldmath$\psi$}_{1\max}(\infty)|^{2}}{\frac{1}{P}+\alpha\|\mathbf{H_{3}}\mbox{\boldmath$\psi$}_{1\max}(\infty)\|^{2}}\\
\nonumber&=\lambda_{\max}\left(\mathbf{h_{1}}\mathbf{h_{1}^{\dag}},\mathbf{H_{3}^{\dag}}\mathbf{H_{3}}\right)
\end{IEEEeqnarray}
As the lower bound and upper bound coincide then we obtain
(\ref{tmp1}). Similarly to obtain (\ref{tmp2}) note that since
\begin{IEEEeqnarray}{rl}\nonumber
\lambda_{(2,2)\max}(\alpha,P)=\frac{1+\frac{(1-\alpha)P|\mathbf{h_{2}^{\dag}}\mbox{\boldmath$\psi$}_{2\max}(\alpha,P)|^{2}}{1+\alpha
P|\mathbf{h_{2}^{\dag}}\mbox{\boldmath$\psi$}_{1\max}(\alpha,P)|^{2}}}{1+\frac{(1-\alpha)P\|\mathbf{H_{3}}\mbox{\boldmath$\psi$}_{2\max}(\alpha,P)\|^{2}}{1+\alpha
P\|\mathbf{H_{3}}\mbox{\boldmath$\psi$}_{1\max}(\alpha,P)\|^{2}}}>1
\end{IEEEeqnarray}
where
\begin{equation}\nonumber
\mbox{\boldmath$\psi$}_{2\max}(\alpha,P)=\arg
\max_{\left\{\mbox{\boldmath$\psi$}_{2}:\|\mbox{\boldmath$\psi$}_{2}\|^{2}=1\right\}}\frac{1+\frac{(1-\alpha)P|\mathbf{h_{2}^{\dag}}\mbox{\boldmath$\psi$}_{2\max}(\alpha,P)|^{2}}{1+\alpha
P|\mathbf{h_{2}^{\dag}}\mbox{\boldmath$\psi$}_{1\max}(\alpha,P)|^{2}}}{1+\frac{(1-\alpha)P\|\mathbf{H_{3}}\mbox{\boldmath$\psi$}_{2\max}(\alpha,P)\|^{2}}{1+\alpha
P\|\mathbf{H_{3}}\mbox{\boldmath$\psi$}_{1\max}(\alpha,P)\|^{2}}}
\end{equation}
for all $P>0$ we have,
\begin{equation}\nonumber
\frac{(1-\alpha)P|\mathbf{h_{2}^{\dag}}\mbox{\boldmath$\psi$}_{2\max}(\alpha,P)|^{2}}{1+\alpha
P|\mathbf{h_{2}^{\dag}}\mbox{\boldmath$\psi$}_{1\max}(\alpha,P)|^{2}}>\frac{(1-\alpha)P\|\mathbf{H_{3}}\mbox{\boldmath$\psi$}_{2\max}(\alpha,P)\|^{2}}{1+\alpha
P\|\mathbf{H_{3}}\mbox{\boldmath$\psi$}_{1\max}(\alpha,P)\|^{2}}
\end{equation}
Therefore, we have
\begin{IEEEeqnarray}{rl}
\lim_{P\rightarrow\infty}\lambda_{(2,2)\max}(\alpha,P)&\leq\frac{\frac{|\mathbf{h_{2}^{\dag}}\mbox{\boldmath$\psi$}_{2\max}(\infty)|^{2}}{|\mathbf{h_{2}^{\dag}}\mbox{\boldmath$\psi$}_{1\max}(\infty)|^{2}}}{\frac{\|\mathbf{H_{3}}\mbox{\boldmath$\psi$}_{2\max}(\infty)\|^{2}}{\|\mathbf{H_{3}}\mbox{\boldmath$\psi$}_{1\max}(\infty)\|^{2}}}\\
\nonumber&\leq
\frac{\lambda_{\max}\left(\mathbf{h_{2}}\mathbf{h_{2}^{\dag}},\mathbf{H_{3}^{\dag}}\mathbf{H_{3}}\right)}{b}
\end{IEEEeqnarray}
where
\begin{IEEEeqnarray}{rl}\nonumber
b=\frac{|\mathbf{h_{2}^{\dag}}\mbox{\boldmath$\psi$}_{1\max}(\infty)|^{2}}{\|\mathbf{H_{3}}\mbox{\boldmath$\psi$}_{1\max}(\infty)\|^{2}}\\
\nonumber \mbox{\boldmath$\psi$}_{2}(\infty)=\arg
\max_{\left\{\mbox{\boldmath$\psi$}_{2}:\|\mbox{\boldmath$\psi$}_{2}\|^{2}=1\right\}}\frac{|\mathbf{h_{2}^{\dag}}\mbox{\boldmath$\psi$}_{2\max}|^{2}}{\|\mathbf{H_{3}}\mbox{\boldmath$\psi$}_{2\max}\|^{2}}
\end{IEEEeqnarray}
On the other hand we have the following lower bound
\begin{IEEEeqnarray}{rl} \label{eq18}
\lim_{P\rightarrow\infty}\lambda_{(2,2)\max}(\alpha,P)\geq\frac{1+\frac{(1-\alpha)|\mathbf{h_{2}^{\dag}}\mbox{\boldmath$\psi$}_{2\max}(\infty)|^{2}}{\alpha|\mathbf{h_{2}^{\dag}}\mbox{\boldmath$\psi$}_{1\max}(\infty)|^{2}}}{1+\frac{(1-\alpha)\|\mathbf{H_{3}}\mbox{\boldmath$\psi$}_{2\max}(\infty)\|^{2}}{\alpha\|\mathbf{H_{3}}\mbox{\boldmath$\psi$}_{1\max}(\infty)\|^{2}}}
\end{IEEEeqnarray}
 Note that $0\leq
\alpha \leq 1$. It is easy to show that the right side of equation
of (\ref{eq18}) is a decreasing function of $\alpha$ and therefore
the maximum value of this function is when $\alpha=0$. Thus we have,
\begin{IEEEeqnarray}{rl}\nonumber
\lim_{P\rightarrow\infty}\lambda_{(2,2)\max}(\alpha,P)\geq
\frac{\lambda_{\max}\left(\mathbf{h_{2}}\mathbf{h_{2}^{\dag}},\mathbf{H_{3}^{\dag}}\mathbf{H_{3}}\right)}{b}
\end{IEEEeqnarray}
As the lower bound and upper bound coincide then we obtain
(\ref{tmp2}). When $\pi={2,1}$, the proof is similar and may be
omitted here.
\end{proof}
Now consider the MISOME channel with $m$ single antenna receivers
and an external eavesdropper. Let
$\mathbf{x}=\sum_{k=1}^{m}\mathbf{b_{k}}$, where
$\mathbf{b_{k}}=u_{k}\psi_{k\max}$, $u_{k}\sim
\mathcal{N}(0,\alpha_{k}P)$, and $\sum_{k=1}^{m}\alpha_{k}=1$.
Assume that $(\lambda_{k\max},\mbox{\boldmath$\psi$}_{k\max})$ is
the largest generalized eigenvalue and the corresponding eigenvector
pair of the pencil
\begin{equation}\nonumber
\left(\mathbf{I}+\frac{\alpha_{k}P
\mathbf{h_{k}}\mathbf{h_{k}^{\dag}}}{1+\mathbf{h_{k}^{\dag}}\mathbf{A}\mathbf{h_{k}}},\mathbf{I}+\alpha_{k}P\mathbf{H_{3}^{\dag}}\left(\mathbf{I}+\mathbf{H_{3}}\mathbf{A}\mathbf{H_{3}^{\dag}}\right)^{-1}\mathbf{H_{3}}\right)
\end{equation}
where
$\mathbf{A}=(\sum_{i=1}^{\pi^{-1}(k)-1}\alpha_{\pi(i)}P\psi_{\pi(i)\max}\psi^{\dag}_{\pi(i)\max})$.
The following corollary then characterizes the capacity region of
the MISOME channel with $m$ receivers under a total power constraint
$P$.

\begin{coro}
Let $\prod$ be the collection of all possible permutations of the
ordered set $\{1,...,m\}$ and $conv$ be the convex closure operator,
then $\mathcal{C}^{MISOME}$ is given as follows.
\begin{equation}\nonumber
\mathcal{C}^{MISOME}=conv\left\{\bigcup_{\pi\in\prod}\mathcal{R}^{MISOME}(\pi)\right\}
\end{equation}
 where $\mathcal{R}^{MISOME}(\pi)$ is given as follows.
\begin{IEEEeqnarray}{rl}\nonumber
\mathcal{R}^{MISOME}(\pi)=\bigcup_{0\leq \alpha_{k} \leq
1,\sum_{k=1}^{m}\alpha_{k}=1}\mathcal{R}^{MISOME}(\pi,\alpha_{1},...,\alpha_{m})
\end{IEEEeqnarray}
where $\mathcal{R}^{MISOME}(\pi,\alpha_{1},...,\alpha_{m})$ is the
set of all $(R_{1},...,R_{m})$ satisfying the following condition.
\begin{IEEEeqnarray}{rl}\nonumber
    R_{k}&\leq \frac{1}{2}\left[\log\lambda_{k\max}\right]^{+}, ~~~~~~~~
    k=1,...,m.
\end{IEEEeqnarray}
\end{coro}

\section{Conclusion} A scenario where a source
node wishes to broadcast two confidential messages for two
respective receivers via a Gaussian MIMO broadcast channel, while a
wire-tapper also receives the transmitted signal via another MIMO
channel is considered. We considered the secure vector Gaussian
degraded  broadcast channel and established its capacity region. Our
achievability scheme was the secret superposition of Gaussian codes.
Instead of solving a nonconvex problem, we used the notion of an
enhanced channel to show that secret superposition of Gaussian codes
is optimal. To characterize the secrecy capacity region of the
vector Gaussian degraded broadcast channel, we only enhanced the
channels for the legitimate receivers, and the channel of the
eavesdropper remained unchanged. Then we extended the result of the
degraded case to non-degraded case. We showed that the secret
superposition of Gaussian codes along with successive decoding
cannot work when the channels are not degraded. we developed a
Secret Dirty Paper Coding (SDPC) scheme and showed that SDPC is
optimal for this channel. Finally, We investigated practical
characterizations for the specific scenario in which the transmitter
and the eavesdropper can afford multiple antennas, while both
intended receivers have a single antenna. We characterized the
secrecy capacity region in terms of generalized eigenvalues of the
receivers' channels and the eavesdropper channel. In high SNR we
showed that the capacity region is a convex closure of two
rectangular regions.
\appendix
\subsection{Equivocation Calculation}
The proof of secrecy requirement for each individual message (\ref{l1}) and (\ref{l2}) is
straightforward and may therefore be omitted.

To prove the requirement of (\ref{l3}) from $H(W_{1},W_{2}|Z^{n})$,
we have
\begin{eqnarray}\nonumber \label{l4}
nR_{e12}&=&H(W_{1},W_{2}|Z^{n})\\
\nonumber  &=& H(W_{1},W_{2},Z^{n})-H(Z^{n})\\
\nonumber &=& H(W_{1},W_{2},V_{1}^{n},V_{2}^{n},Z^{n})- H(V_{1}^{n},V_{2}^{n}|W_{1},W_{2},Z^{n})-H(Z^{n})\\
\nonumber &=&H(W_{1},W_{2},V_{1}^{n},V_{2}^{n})+ H(Z^{n}|W_{1},W_{2},V_{1}^{n},V_{2}^{n})- H(V_{1}^{n},V_{2}^{n}|W_{1},W_{2},Z^{n})-H(Z^{n})\\
\nonumber &\stackrel{(a)}{\geq}&
H(W_{1},W_{2},V_{1}^{n},V_{2}^{n})+H(Z^{n}|W_{1},W_{2},V_{1}^{n},V_{2}^{n})-n\epsilon_{n}-H(Z^{n})\\
\nonumber
&\stackrel{(b)}{=}&H(W_{1},W_{2},V_{1}^{n},V_{2}^{n})+H(Z^{n}|V_{1}^{n},V_{2}^{n})-n\epsilon_{n}-H(Z^{n}) \\
\nonumber&\stackrel{(c)}{\geq}&
H(V_{1}^{n},V_{2}^{n})+H(Z^{n}|V_{1}^{n},V_{2}^{n})- n\epsilon_{n}-H(Z^{n}) \\
\nonumber &=& H(V_{1}^{n})+H(V_{2}^{n})-I(V_{1}^{n};V_{2}^{n})- I(V_{1}^{n},V_{2}^{n};Z^{n})- n\epsilon_{n}\\
\nonumber &\stackrel{(d)}{\geq}& I(V_{1}^{n};Y_{1}^{n})+I(V_{2}^{n};Y_{2}^{n})-I(V_{1}^{n};V_{2}^{n})-I(V_{1}^{n},V_{2}^{n};Z^{n})- n\epsilon_{n}\\
\nonumber &\stackrel{(e)}{\geq}& nR_{1}+nR_{2}-n\epsilon_{n},
\end{eqnarray}
where $(a)$ follows from Fano's inequality, which states that for
sufficiently large $n$, $H(V_{1}^{n},V_{2}^{n}|W_{1},W_{2},Z^{n})$
$\leq h(P_{we}^{(n)})$ $+nP_{we}^{n}R_{w}\leq n\epsilon_{n}$. Here
$P_{we}^{n}$ denotes the wiretapper's error probability of decoding
$(v_{1}^{n},v_{2}^{n})$ in the case that the bin numbers $w_{1}$ and
$w_{2}$ are known to the eavesdropper. Since the sum rate is small
enough, then $P_{we}^{n}\rightarrow 0$ for sufficiently large $n$.
$(b)$ follows from the following Markov chain:
$(W_{1},W_{2})\rightarrow (V_{1}^{n},V_{2}^{n})\rightarrow$ $Z^{n}$.
Hence, we have
$H(Z^{n}|W_{1},W_{2},V_{1}^{n},V_{2}^{n})=H(Z^{n}|V_{1}^{n},V_{2}^{n})$.
$(c)$ follows from the fact that
$H(W_{1},W_{2},V_{1}^{n},V_{2}^{n})\geq H(V_{1}^{n},V_{2}^{n})$.
$(d)$ follows from that fact that $H(V_{1}^{n})\geq
I(V_{1}^{n};Y_{1}^{n})$ and $H(V_{2}^{n})\geq
I(V_{2}^{n};Y_{2}^{n})$. $(e)$ follows from the following lemmas.

\begin{lem}\label{lem3}
Assume $V_{1}^{n},V_{2}^{n}$ and $Z^{n}$ are generated according to
the achievablity scheme of Theorem \ref{t4}, then we have,
\begin{IEEEeqnarray}{lr}\nonumber
I(V_{1}^{n},V_{2}^{n};Z^{n})\leq nI(V_{1},V_{2};Z)+n\delta_{1n},\\
\nonumber I(V_{1}^{n};V_{2}^{n})\leq nI(V_{1};V_{2})+n\delta_{2n}.
\end{IEEEeqnarray}
\end{lem}
\begin{proof}
Let $A_{\epsilon}^{n}(P_{V_{1},V_{2},Z})$ denote the set of typical
sequences $(V_{1}^{n},V_{2}^{n},Z^{n})$ with respect to
$P_{V_{1},V_{2},Z}$, and
\begin{equation}\nonumber
\zeta=\left\{
       \begin{array}{ll}
         1, & (V_{1}^{n},V_{2}^{n},Z^{n})\notin A_{\epsilon}^{n}(P_{V_{1},V_{2},Z}); \\
         0, & \hbox{otherwise},
       \end{array}
     \right.
\end{equation}
be the corresponding indicator function. We expand
$I(V_{1}^{n},V_{2}^{n};Z^{n})$ as follow,
\begin{IEEEeqnarray}{rl}\label{l31}
I(V_{1}^{n},V_{2}^{n};Z^{n})&\leq
I(V_{1}^{n},V_{2}^{n},\zeta;Z^{n})\\ \nonumber
&=I(V_{1}^{n},V_{2}^{n};Z^{n},\zeta)+I(\zeta;Z^{n})\\
\nonumber
&=\sum_{j=0}^{1}P(\zeta=j)I(V_{1}^{n},V_{2}^{n};Z^{n},\zeta=j)+I(\zeta;Z^{n}).
\end{IEEEeqnarray}
According to the joint typicality property, we have
\begin{IEEEeqnarray}{rl}\label{l32}
P(\zeta=1)I(V_{1}^{n},V_{2}^{n};Z^{n}|\zeta=1)&\leq
nP((V_{1}^{n},V_{2}^{n},Z^{n})\notin
A_{\epsilon}^{n}(P_{V_{1},V_{2},Z}))\log\|\mathcal{Z}\|\\
\nonumber &\leq n\epsilon_{n}\log\|\mathcal{Z}\|.
\end{IEEEeqnarray}
Note that,
\begin{equation}\label{l33}
I(\zeta;Z^{n})\leq H(\zeta)\leq 1
\end{equation}
Now consider the term
$P(\zeta=0)I(V_{1}^{n},V_{2}^{n};Z^{n}|\zeta=0)$. Following the
sequence joint typicality properties, we have
\begin{IEEEeqnarray}{rl}\label{l34}
P(\zeta=0)I(V_{1}^{n},V_{2}^{n};Z^{n}|\zeta=0)&\leq
I(V_{1}^{n},V_{2}^{n};Z^{n}|\zeta=0)\\ \nonumber &=
\sum_{(V_{1}^{n},V_{2}^{n},Z^{n})\in
A_{\epsilon}^{n}}P(V_{1}^{n},V_{2}^{n},Z^{n})\big(\log
P(V_{1}^{n},V_{2}^{n},Z^{n})-\log
P(V_{1}^{n},V_{2}^{n})\\ \nonumber &-\log P(Z^{n})\big),\\
\nonumber &\leq
n\left[-H(V_{1},V_{2},Z)+H(V_{1},V_{2})+H(Z)+3\epsilon_{n}\right],\\
\nonumber &=n\left[I(V_{1},V_{2};Z)+3\epsilon_{n}\right].
\end{IEEEeqnarray}
By substituting (\ref{l32}), (\ref{l33}), and (\ref{l34}) into
(\ref{l31}), we get the desired reasult,
\begin{IEEEeqnarray}{rl}
I(V_{1}^{n},V_{2}^{n};Z^{n})&\leq
nI(V_{1},V_{2};Z)+n\left(\epsilon_{n}\log\|\mathcal{Z}\|+3\epsilon_{n}+\frac{1}{n}\right),\\
\nonumber &=nI(V_{1},V_{2};Z)+n\delta_{1n},
\end{IEEEeqnarray}
where,
\begin{equation}\nonumber
\delta_{1n}=\epsilon_{n}\log\|\mathcal{Z}\|+3\epsilon_{n}+\frac{1}{n}.
\end{equation}
Following the same steps, one can prove that
\begin{IEEEeqnarray}{lr}
I(V_{1}^{n};V_{2}^{n})\leq nI(V_{1};V_{2})+n\delta_{2n}.
\end{IEEEeqnarray}
\end{proof}
Using the same approach as in Lemma \ref{lem3}, we can prove the
following lemmas.
\begin{lem}\label{lem4}
Assume $ V_{1}^{n}, Y_{1}^{n}$ and $Y_{2}^{n}$ are generated
according to the achievablity scheme of Theorem \ref{t4}, then we
have,
\begin{IEEEeqnarray}{lr}\nonumber
I(V_{1}^{n};Y_{1}^{n})\leq nI(V_{1};Y_{1})+n\delta_{3n},\\
\nonumber I(V_{2}^{n};Y_{2}^{n})\leq nI(V_{1};Z)+n\delta_{4n}.
\end{IEEEeqnarray}
\end{lem}
\begin{proof}
The steps of the proof are very similar to the steps of proof of
Lemma \ref{lem3} and may be omitted here.
\end{proof}

\end{document}